\documentclass[aps,prl,reprint,superscriptaddress]{revtex4-2}

\usepackage[utf8]{inputenc}
\usepackage[pdfpagelabels]{hyperref}

\usepackage{mathrsfs}
\usepackage{amsmath}
\usepackage{amssymb}
\usepackage{amsfonts}
\usepackage{amsthm,bbm}
\usepackage{graphicx}

\newcounter{myequation}
\makeatletter
\@addtoreset{equation}{myequation}
\makeatother

\usepackage{hyperref}
\hypersetup{colorlinks=true,allcolors=[rgb]{0,0,0.6}}

\newtheorem{theorem}{Theorem}
\newtheorem{lemma}[theorem]{Lemma}
\newtheorem{proposition}[theorem]{Proposition}

\newtheorem{definition}[theorem]{Definition}

\usepackage{cancel}
\usepackage[normalem]{ulem}
\usepackage{soul}
\setstcolor{blue}

\usepackage{color}
 \definecolor{green}{rgb}{0.0, 0.5, 0.5}
\newcommand{\R}{\mathbb{R}}
\newcommand{\C}{\mathbb{C}}
\newcommand{\Z}{\mathbb{Z}}

\newcommand{\cL}{\mathcal{L}}
\newcommand{\astlo}{\mathbb A}
\newcommand{\nc}{\newcommand}

\nc{\al}{\alpha}
\nc{\del}{\delta}
\nc{\h}{\delta}
\nc{\Gam}{\Gamma}
\nc{\G}{\Gamma}
\nc{\et}{\eta} 
\nc{\g}{\gamma}
\nc{\gam}{\gamma}
\nc{\ka}{\kappa}
\nc{\lam}{\lambda}
\nc{\Lam}{\Lambda}
\nc{\Om}{\Omega}
\nc{\om}{\omega}

\nc{\ta}{\tau}
\nc{\w}{\omega}
\nc{\io}{\iota}
\nc{\z}{\zeta}
\nc{\s}{\sigma}
\nc{\Si}{\Sigma}
\nc{\vphi}{\varphi}

\nc{\e}{\epsilon}

\newcommand{\eps}{\epsilon}

\nc{\ran}{\rangle}
\nc{\lan}{\langle}

\newcommand{\re}{\operatorname{Re}}
\newcommand{\im}{{\rm Im}}

\nc{\bfone}{{\bf 1}}
\nc{\1}{{\bf 1}}

\newcommand{\beq}{\begin{equation}}
\newcommand{\eeq}{\end{equation}}

\nc{\dd}{\mathrm{d}}

\newcommand{\DETAILS}[1]{}

\newcommand{\x}{|x|}
\newcommand{\y}{|y|}

\begin{document}

 \title{Maximal speed for macroscopic particle transport in the Bose-Hubbard model}%[Maximal velocity estimates]

 \author{J\'er\'emy Faupin} \affiliation{Institut Elie Cartan de Lorraine, Université de Lorraine, 57045 Metz Cedex 1, France} \email{jeremy.faupin@univ-lorraine.fr}
 
 \author{Marius Lemm} \affiliation{FB Mathematik, Universit\"at Tübingen,
72076 T\"ubingen, Germany} \email{marius.lemm@uni-tuebingen.de}
 
 \author{Israel Michael Sigal} \affiliation{Department of Mathematics, University of Toronto, Toronto, M5S 2E4, Ontario, Canada} \email{im.sigal@utoronto.ca}

 \begin{abstract}
The Lieb-Robinson bound asserts the existence of a maximal propagation speed for the quantum dynamics of lattice spin systems. Such general bounds are not available for most bosonic lattice gases due to their unbounded local interactions. Here we establish for the first time a general ballistic upper bound on macroscopic particle transport in the paradigmatic Bose-Hubbard model. The bound is the first to cover a broad class of initial states with positive density including Mott states, which resolves a longstanding open problem. It applies to Bose-Hubbard type models on any lattice with not too long-ranged hopping. The proof is rigorous and rests on controlling the time evolution of a new kind of adiabatic spacetime localization observable via iterative differential inequalities. 
	\end{abstract}
	
		 \maketitle
		% \maketitle

%%%%%%%%%%%%%%%%%%%%%%%%%%%%%%%%%%%%%%%%%%%%%%%%%%%%%%%%%%%%%%%%%%%%%%%%%
%\begin{document}

%\hypersetup{pageanchor=false}

%\maketitle

%\hypersetup{pageanchor=true}

%\pagenumbering{arabic}
%\bibliographystyle{abbrv}
A central tenet of relativistic theory is the existence of the light cone, i.e., an absolute upper bound on the speed of propagation. It is a remarkable fact that many non-relativistic condensed-matter systems similarly display an \textit{effective ``light'' cone} which provides a system-dependent upper bound on the maximal speed of quantum propagation. In contrast to its relativistic counterpart, this effective light cone leaks exponentially small errors as is typically unavoidable in quantum dynamics.
This deep fact was discovered by Lieb and Robinson \cite{LR} for quantum spin systems on lattices. The resulting \textit{Lieb-Robinson bound} showed that the ultraviolet cutoff imposed by the lattice provides a maximal speed of propagation on the many-body dynamics. The interest in Lieb-Robinson bounds rapidly surged in the early 2000s when it became clear that they are among the very few effective and general tools that are available for analyzing quantum many-body systems. Accordingly, they have played a decisive role in contexts as diverse as quantum information science \cite{H1,LVV}, condensed-matter theory \cite{BdRF,BMNS,BHM,BHV,H2,NS_ls} and high-energy physics \cite{CL,KS_he,RS} to name a few.

%For instance, they are the key tool in Hastings' famous derivation of the area law for gapped quantum spin chains \cite{H}, a milestone of quantum information theory with major consequences for quantum complexity theory \cite{LVV}. There are many other problems in which LR bounds play a decisive role, e.g., the classification of topological quantum phases \cite{BMNS,BHM} and estimates on the generation time for topological order \cite{BHV}. 
A variety of improvements of the original Lieb-Robinson bound have been achieved over the past 10 years \cite{DLLY1,DLLY2,EMNY,Fossetal,GL,GNRS,HSS,H2,MKN,NRSS,NS1,NSY1,Tranetal} including, e.g., extensions to long-range spin interactions and fermionic lattice gases.
For a more complete discussion, see the survey papers \cite{KGE,NS2,NSY2}.

Despite these celebrated successes, a nagging limitation of the Lieb-Robinson bounds has persisted over the years---the standard proofs are fundamentally limited to \textit{bounded interactions} as enjoyed by quantum spin systems. Certain oscillator systems with unbounded interactions have been addressed by different methods \cite{NRSS}. However, for general unbounded interactions, the standard arguments only yield an unsatisfactory bound on the maximal speed which is proportional to the total particle number $N$, a trivial bound in the thermodynamic limit. 

This limitation largely leaves out the wide field of \textit{bosonic quantum lattice gases} since these naturally come with unbounded interactions, for example the paradigmatic \textit{Bose-Hubbard (BH)  model}  \cite{BDZ}. Experiments with ultracold gases and numerical simulations have found an effective light cone for the BH model after a quench \cite{exp1,exp2,nexp1,nexp2}. On the theoretical side, a fully satisfactory understanding of this fact is lacking. It is known that the problem is subtle because superballistic transport can occur in certain related examples \cite{EG}. 

A small number of theoretical results have established a maximal propagation speed for bosonic lattice gases for special initial states. A first maximal speed bound in the BH model was given in \cite{SHOE} for initial states that have no particles outside of a fixed region. This condition excludes states of positive local density, e.g., Mott states \eqref{eq:mott}. Very recently, a number of groups have made progress on this problem through novel techniques: The $N$-scaling of the velocity was improved to $\sqrt{N}$ \cite{WH}; an almost-linear light cone was derived for special initial states that are local perturbations of a stationary state satisfying certain exponential constraints on the local particle density \cite{KS}; a linear light cone was derived for commutators tested against the state $e^{-\mu N}$ \cite{YL}; and \cite{SHOE} was extended to propagation through vacuum \cite{FLS}. 

In this Letter, we show for the first time the \textit{finiteness of the speed of macroscopic particle transport in the BH model for general initial states}. We obtain an explicit bound \eqref{eq:vmax} on the maximal speed that is independent of the particle number and easily computable from the hopping parameters of the Hamiltonian. 
  In particular, our result is the first to provide a thermodynamically stable ballistic particle propagation bound on the prototypical Mott states \eqref{eq:mott} which resolves a longstanding open problem. See Theorem \ref{thm:main} below for the formal statement. Our result 
is a new kind of \textit{macroscopic-type Lieb-Robinson bound} for particle transport. It remains to be seen if the method can be adapted to propagation of other physical characteristics, %quantities,
  e.g., entanglement.

Our main idea is to control the time evolution by means of a new class of observables which we call \textit{adiabatic spacetime localization observables (ASTLO)}. The construction is strongly inspired by the method of propagation observables developed in \cite{APSS,BonyFaupSig,FaupSig,HeSk,HS,SigSof,Skib} and thereby connects these developments to the study of many-body lattice gases for the first time.

The %adiabatic nature of the 
specially designed ASTLOs track %allows us to control
 in a precise way how
 the many-body system dynamically spreads in spacetime,
 while decreasing along quantum evolution. The latter key property allows one to convert the ASTLOs' spacetime localization into suitable estimates on the propagator. This is proven  through iterative differential inequalities obtained by %iterated
Taylor series-like commutator expansion. These techniques are fully analytical, rigorous, and robust. Accordingly, the proof applies to a wide variety of BH type models with rather long-ranged hopping and on general lattices.

\section{Setting and main result} 

%\subsection{Bose-Hubbard models}
We consider a finite subset $\Lam$ of a lattice $\cL\subset \R^d$. For example, $\cL=\Z^d$ and $\Lam$ is a discrete box. We shall prove bounds that are independent of the number of sites in $\Lambda$ and which therefore extend to the infinite-volume limit.

We consider a system of bosons on $\Lam$ described by the generalized Bose-Hubbard model Hamiltonian
\begin{align}\label{H-Lam}H_\Lam=-\sum_{x, y\in\Lam} J_{xy}^\Lam b_x^\dagger b_y+ \sum_{x\in\Lam}V_x(n_x)- \mu\sum_{x\in\Lam}  n_x.\end{align} 
acting on the bosonic Fock space $\mathcal F$. 

We assume that $J_{x,y}^\Lam=J_{y,x}^\Lam$ and we let $V_x:\{0,1,2,\dots\}\to\R$ be \textit{an arbitrary local potential.} 

The standard BH Hamiltonian involves nearest-neighbor hopping and quadratic on-site interaction \cite[eq.\  (65)]{BDZ}, i.e.,
\begin{equation}\label{eq:standardBH}
J_{x,y}^\Lam=J\delta_{\substack{x\sim_\Lam y}},
\qquad V_x(n_x)=V(n_x)=\frac{U}{2} n_x(n_x-1).
\end{equation}
where $x\sim_\Lam y$ means $x$ and $y$ are nearest neighbors in $\Lambda$, possibly subject to periodic boundary conditions if desired.

We allow for long-ranged hopping in the BH Hamiltonian. The hopping range is quantified by an integer parameter $p$ and the quantity
   \begin{align}\label{kaJp-def} %{J-cond}
	 \kappa_J^{(p)}=\max_{x\in\Lam}\sum_{y\in\Lam} |J_{xy}^\Lam| |x-y|^{p}
	  \end{align} 
	  where $|\cdot|$ denotes the Euclidean distance.
Our bounds will involve the constant $\kappa_J^{(p)}$ for some $p\geq 2$ and to have a well-defined infinite-volume limit, we are interested in situations where $\kappa_J^{(p)}$  is bounded \textit{independently of $\Lam$}. For example, if we consider $\Lam\subset \Z^d$ and $|J^\Lam_{xy}| \lesssim | x-y|^{-\alpha}$ for some exponent $\alpha\geq d+1$, then $\kappa_J^{(\alpha-d)}$ is independent of $\Lam$. For finite-range (or exponentially decaying) hopping, we can take $p$ arbitrarily large. 

We will show that the \textit{maximal propagation speed} is given by
\beq\label{eq:vmax}
v_{\max }\equiv  \kappa^{(1)}_J=\max_{x\in\Lam}\sum_{y\in\Lam} |J_{xy}^\Lam| |x-y|.
\eeq
For nearest-neighbor hopping $J_{xy}^{\Lambda}=J\delta_{x\sim y}$, we have $v_{\max}= J\max_{x}\mathrm{deg}(x)$ assuming the lattice embedding is such that nearest neighbors have Euclidean distance $1$.

%REMARK: Long-range hoppings with $p<d$ do not display a linear light cone for general information propagation \cite{}, but here we focus on particle propagation. 

%In Appendix ... \ref{sec:H-gen-prop} we prove for the sake of completeness that the operator $H_\Lambda$ is self-adjoint on the appropriate domain.

%\subsection{Main result}
Our main result controls the macroscopic change of local particle numbers outside of an effective light cone with slope determined by $v_{\max}$. To formulate it precisely, we define for a given subset $S\subset \Lam$, the local particle numbers
\begin{equation} \label{eq:Nlocdefn}
N_S=\sum_{x\in S}n_x,\qquad \bar N_S=\frac{N_S}{N_\Lam}.
\end{equation}

We recall that the total particle number $N_\Lam=\sum_{x\in \Lam} n_x$ is conserved by $H_\Lam$. For $c\in\R $ and $S\subset \Lam$, we write $P_{\bar N_S< c}$, $P_{\bar N_{S^c}\geq c}$, etc., for the associated spectral projectors of $\bar N_S$, where $S^c=\Lambda\setminus S$. 

%We fix a smooth cutoff function $\chi:(0,\infty)\to [0,1]$ which is equal to $0$ on $[0,1\2]$ and equal to 

Given a set $S\subset \Lambda$, we write $R_{\min}(S)$ for the radius of the smallest Euclidean ball $B$ so that $S\subseteq B$.  %Finally, we write $\tau_t(A)= e^{iHt}A e^{-iHt}$ for the Heisenberg evolution of a bounded observable $A$.
We write $\lan A\ran_\psi=\lan \psi, A\psi\ran$ for the expectation value of an observable $A$ in state $\psi$. Given two subsets of the lattice $X,Y\subset \Lam$, we write $d_{XY}$ for their Euclidean distance.
%\bigskip

	\noindent\begin{figure}
\begin{center}
    \includegraphics[scale=.5]{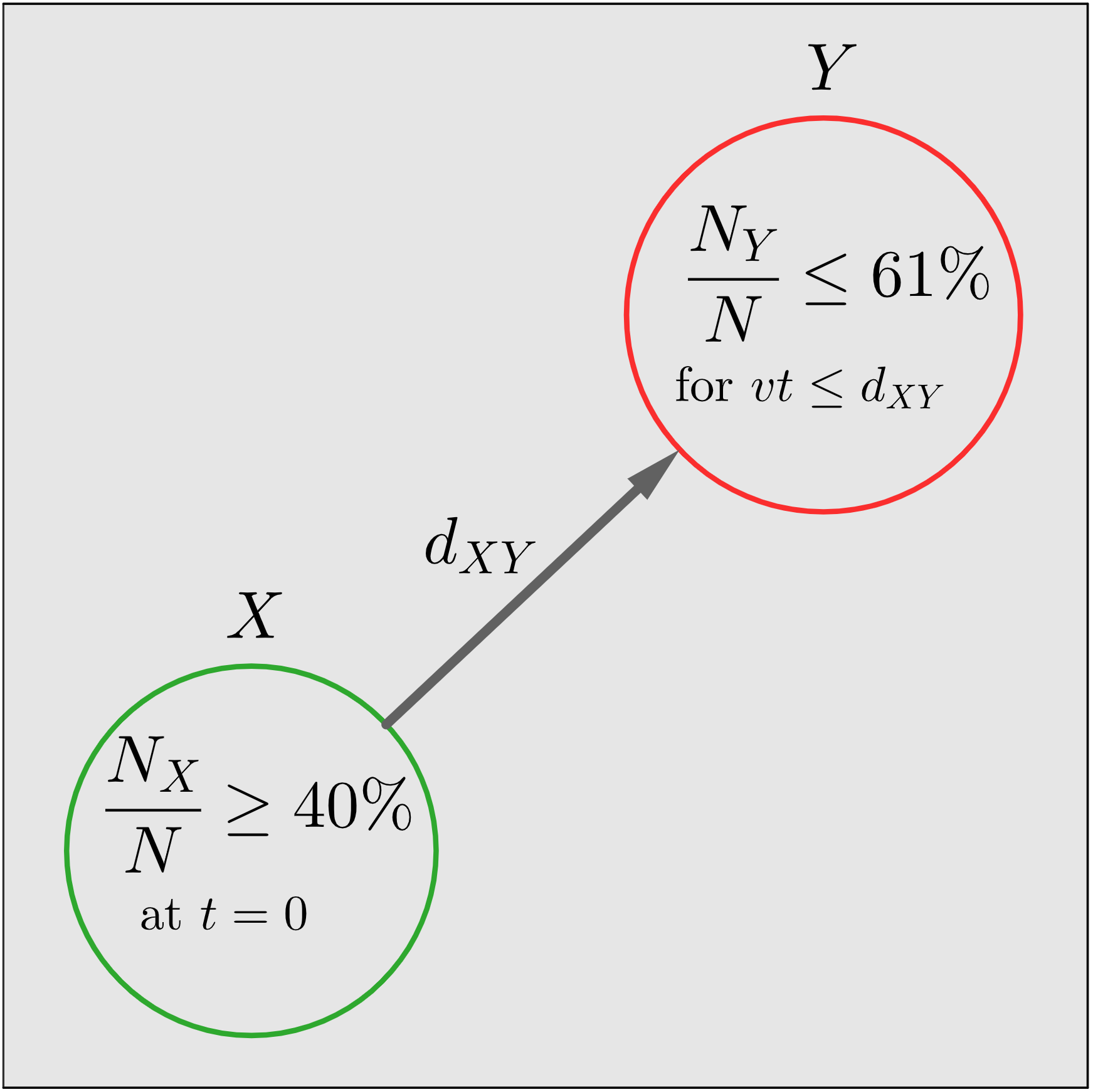}
    \caption{As shown in Theorem \ref{thm:main}, the transport of $1\%$ of the particles from $X$ to $Y$ takes time proportional to $d_{XY}$. A macroscopic cloud of particles moves at most at speed $v_{\max}$.}
    \label{sketch}
\end{center}
\end{figure}

%%%%%%%%%%%%%
\begin{theorem}[Main result]\label{thm:MaxVE-Hub}\label{thm:main} 	Consider the Hamiltonian $H_\Lam$ given by \eqref{H-Lam} with the hopping matrix $J_{x y}^{\Lambda}$ satisfying %\eqref{J-cond}
 $\kappa_J^{(p)}<\infty$ for some $p\geq 2$.  Fix numbers $v>v_{\max }$ and $0\leq \eta<\xi\leq 1$.

Let $X$ and $Y$ be disjoint subsets of $\Lambda$ and let $\phi$ be any normalized state. Consider the time-evolved state
\begin{equation} \label{eq:psitdefn}
\psi_t=	e^{-itH} P_{ \bar N_{X^c}\leq \eta}\phi.
 \end{equation}
Then we have the decay estimate 
\begin{align} \label{max-vel-estHub'}
	\lan P_{ \bar N_{Y}\ge \xi}\ran_{\psi_t}	\le C_{\kappa_J^{(p)}} \,d_{XY}^{1-p},		\end{align}
whenever $d_{XY}\geq v t+2R_{\min}(X)$.
%whenever \beq\label{eq:lightconecondn} d_{XY}\geq v t+2R_{\min}(X). \eeq
\end{theorem}

To interpret the result, see Figure \ref{sketch} and consider an initial state $\phi$ so that $P_{\bar N_{X^c} \le \eta} \phi=\phi$, meaning the fraction of particles outside of a ball $X$ is at most $\eta$ (say, $\eta=0.6$ and so at least $60 \%$ of all particles are outside of $X$). Then \eqref{max-vel-estHub'} shows that the time it takes to raise the fraction of particles inside $Y$ to $\xi>\eta$ (say, to $61 \%$ of all particles) is at least proportional to the distance $d_{XY}$. In short, moving $(\xi-\eta)N$ particles from $X$ to $Y$ takes time proportional to $d_{XY}$. This proves that macroscopic many-body transport is at most ballistic. 

A few remarks on Theorem \ref{thm:MaxVE-Hub} are in order. (i) The notation $C_{\kappa_J^{(p)}}$ means that the constant depends on the value of $\kappa_J^{(p)}$. (ii) The left-hand side of \eqref{max-vel-estHub'} vanishes at $t=0$. We prove that it remains small as long as one stays outside of an effective light cone \beq\label{eq:lightconecondn}d_{XY}\geq v t+2R_{\min}(X) \eeq (see \eqref{max-vel-estHub'}).  For finite-range hopping, the decay outside of the effective light cone is faster than any polynomial. (iii) The maximal speed $v_{\max }$ from \eqref{eq:vmax} is independent of particle number and of the observables $X$ and $Y$. It only depends on model parameters similarly to the Lieb-Robinson velocity.
(iv) The result applies to a broad class of initial states including ones that can have positive local particle density. This allows, for the first time, to consider the important class of \textit{Mott states}
\begin{equation}\label{eq:mott}
\phi=\bigotimes_{x\in \Lambda} (a_x^\dagger)^{\nu_x}|0\ran,\qquad \nu_x\in \{0,1,2,\ldots\}.
\end{equation}
(A common choice is $\nu_x\equiv \nu$ with $\nu-1<\frac{\mu}{U}<\nu$ which gives a Mott insulating ground state of \eqref{eq:standardBH} in the limit $U\gg J$.)
(v) The term $2R_{\min}(X)$ in the condition following \eqref{eq:lightconecondn} plays no role when $X$ is a fixed bounded set. Moreover, if $d_{0Y}=d_{XY}+R_{\min }(X)$ (e.g., if $X$ has symmetry) then \eqref{eq:lightconecondn} can be relaxed to $d_{XY}\geq vt$ even if $X$ grows with system size. Finally, the constant $2$ can be replaced by any number $>1$.

%(v) The constant $C$ in Theorem \ref{thm:MaxVE-Hub} depends on the parameters $n,\xi,\eta,\eps$ in  an explicit way \eqref{eq:constant}. From this formula, we see that we can in fact choose $\xi-\eta$ as $N$-dependent. Naturally, the bound is strongest when $\xi-\eta$ is of order $1$, i.e., when a macroscopic fraction of particles (e.g.\ $1 \%$) are considered to propagate from $X$ to $Y$. 

%\section{A related propagation bound}

\section{ASTLOs: Definition and basic properties}
%We summarize the main steps of their construction and refer to the supplemental material for details. 
 %\subsection{ASTLOs: definition and basic properties}
The overarching idea behind our approach is to construct special adiabatic spacetime localization observables (ASTLO) (see \eqref{eq:Ldefn}  below) which decrease monotonically along quantum trajectories (up to inessential fast decaying terms). The proof is based on iterative differential inequalities with the adiabatic nature of ASTLO's playing an important role.

An important feature of the ASTLO construction is that we use smooth, slowly varying (adiabatic) cutoff functions instead of sharp ones. 

Given  $v>v_{\max}$, let $\eps\in (0,\tfrac{1}{2})$ be small enough such that $v'=(1-\eps)v>v_{\max}$ still. We define the smeared out light cone indicator as 
\beq\label{eq:chitdefn}
\chi_{t}(|x|)=\chi\left(\frac{|x|-R_{\min}(X)-v't}{\eps d_{XY}}\right),
\eeq
where $\chi$ is a smoothed out indicator function of the semi-interval $[0, \infty)$; see Figure \ref{fig1} in the supplemental material (SM). (A precise definition will be given below.) By translation, we may assume that $X\subset \Lam$ is contained in $B_{R_{\min}(X)}$, the Euclidean ball of radius $R_{\min}(X)$ centered at $0$. 

We consider $ d_{XY}$ as the large adiabatic parameter that makes $\chi_{t}(x)$ slowly varying. 
The associated adiabatic spacetime localization operator (ASTLO) is then the Fock space operator $\astlo_t$ given by the (normalized) second quantization of $\chi_t$, i.e., 
\begin{equation}\label{eq:Ldefn}
	\astlo_{t}=\frac{1}{N_\Lambda}\sum_{x\in \Lambda} \chi_{t}(|x|) n_x.
\end{equation}
Physically, the ASTLO $\astlo_{t}$ can be thought of as a smeared-out localized relative number operator. It measures how many particles are at least distance $v't$ away from the ball $B_{R_{\min}(X)}$, but it only fully counts the particles whose distance from the light cone is at least of order $\eps d_{XY}$. Conversely, the particles whose distance from the light cone is positive but $\ll \eps d_{XY}$ contribute almost nothing to $\astlo_t$.

The ASTLOs are useful because, in addition to decreasing monotonically along quantum trajectories, they satisfy the following two somewhat competing properties:
(I) They are closely connected to the more sharply varying local particle numbers $N_{X^c}$ and $N_Y$. (II) Their adiabatic nature leads to a slow time evolution. Mathematically, this means that higher commutators are subleading in the small adiabatic parameter $1/d_{XY}$ which enables the iterative commutator expansion.

Let us explain point (I) further. We begin by noting that local particle number operators and ASTLOs are sums of $n_x$'s and thus commute. Then $x\in X\subset B_{R_{\min}(X)}$ implies $\chi_{0}(|x|)=0$ and so we have the operator inequality
\begin{equation}\label{eq:NLt0}
\bar N_{X^c}\geq \astlo_{0}.
\end{equation}

Since $X$ contains the origin, we have for any $y\in Y$ that $|y|\geq  d_{XY}$. The assumption $d_{XY}\geq vt+2R_{\min}(X)$ and our choice of $\eps$ then imply that
$\chi_{t}(|y|)=1$. Hence, we obtain the second operator inequality
\begin{equation}\label{eq:NL}
\bar N_{Y}\leq \astlo_{t}
\end{equation}
which clarifies point (I) above.
%\sout{The relations \eqref{eq:NLt0} and \eqref{eq:NL} clarify the close connection between local particle numbers and ASTLOs as summarized in point $\mathrm{(I)}$ above.} 

\section{Sketch of proof of Theorem \ref{thm:main}}

%\subsection{Smeared out spectral projectors}
In view of point $\mathrm{(II)}$, one might hope to use the fact that the ASTLOs decrease monotonically along quantum trajectories together with relations \eqref{eq:NLt0} and \eqref{eq:NL}  to estimate the quantities \eqref{eq:psitdefn} and \eqref{max-vel-estHub'} appearing in the main result. While this can be done for special initial conditions similarly to %expansion ideas
 \cite{FLS}, this approach does not work in the generality we desire here. %because of the sharp spectral cutoff implicit in $P_{\bar N_{X^c}\leq \eta}$ and $P_{\bar N_{Y}\geq \xi}$.

To treat positive densities, we introduce an augmented ASTLO by taking a monotonic function of $\astlo_t$. Let $f$ be a monotonic smooth cutoff function that goes from $0$ to $1$ between $\eta$ and $\xi$. To be precise, $f$ belongs to the class of cutoff functions $\mathcal C_{\eta,\xi}$ (the formal definition below can be skipped on first reading)
%\begin{widetext}
$$
\begin{aligned}
\mathcal C_{\eta,\xi}
=
&\Big\{ 		
f\in C^\infty(\R_+)\,:\, f,f'\geq0,\, \sqrt{f'}\in C^\infty(\R_+),\\
& f=0 \textnormal{ on } (0,\eta),\, f=1 \textnormal{ on } (\xi,\infty), \mathrm{supp}\, f'\subset (\eta,\xi)
\Big\}.
\end{aligned}
$$
%\end{widetext}
Now we define the approximate spectral projector for the ASTLO via the spectral theorem as
$$
\Phi(t)=f(\astlo_t)=\sum_{\lam\in\mathrm{spec} \astlo_t} f(\lambda) P_{\lambda}(\astlo_t).
$$ 
with $P_\lam(\astlo_t)$ the projector onto the $\lambda$-eigenspace of $\astlo_t$. %As a monotonic function of an ASTLO,  $\Phi(t)$ is itself an ASTLO. 

%We note that different $\Phi(t)$ commute for all times because they are all diagonal in the Mott basis.
 
%\subsection{Dynamical growth estimates on the approximate spectral projector $\Phi(t)$}
The fact that $f\in \mathcal C_{\eta,\xi}$ implies that $\Phi(t)$ is an approximate spectral projector in the sense that
\beq\label{eq:aboverelations}
P_{\bar N_{X^c}\leq \eta}\Phi(0)=0,
 \qquad P_{\bar N_{Y}\geq \xi}=P_{\bar N_{Y}\geq \xi}\Phi(t).
\eeq
We denote $\lan A\ran_t=\lan A\ran_{\psi_t}$. The above relations \eqref{eq:aboverelations} give
\begin{equation}\label{eq:Phiconditions}
\lan \Phi(0)\ran_0=0,\qquad \lan P_{ \bar N_{Y}\ge \xi}\ran_{t}\leq \lan \Phi(t)\ran_{t}.
\end{equation}
As anticipated, we see that the task reduces to controlling the dynamical growth of the function $t\mapsto \lan \Phi(t)\ran_{t}$ governed by the differential equation
\begin{align}
	\label{dt-Heis}
	{\dd \over{\dd t}}\left<\Phi(t)\right>_t =&\lan D\Phi(t)\ran_t,\\ 
\label{Heis-der}\textnormal{where } D \Phi(t)=&\frac{\partial}{\partial t}\Phi(t)+i[H, \Phi(t)].
\end{align}
$D\Phi(t)$ is called the Heisenberg derivative of $\Phi(t)$. 

%The key input at this stage is the following estimate on the Heisenberg derivative. 
%\st{The key features of the space-time observable $\Phi(t)$ are that (i) its support tracks down the surplus of particles outside the effective light cone, (ii)  it changes with $t$ adiabatically (with the adiabatic parameter $s=\eps d_{XY}$) and (iii) it is positive and decreases monotonically (up to a fast decaying term) along the evolution; see} \eqref{Phi-monot} below.
%We abbreviate $s=\eps d_{XY}$.

\begin{theorem}[Bound on the Heisenberg derivative]\label{thm:timederivmaintxt}
Let $f\in \mathcal C_{\eta,\xi}$ and $\chi\in \mathcal C_{1/2,1}$. Then, there exists a constant $C>0$ and cutoff functions $\tilde f\in \mathcal C_{\eta,\xi}$ and $\tilde \chi\in\mathcal C_{1/2,1}$ such that for all $t$ and all sufficiently large $s$,
	\beq\label{eq:DPhiest}
	\begin{aligned}
	D\Phi(t)
	\leq &-\frac{v'-v_{\max}}{s} f'(\astlo_{t})\astlo'_{t}+\frac{C}{s^{2}} \tilde f'(\tilde \astlo_{t})\tilde \astlo_{t}'+\frac{C}{s^{p}}.
\end{aligned}
	\eeq
		\end{theorem}
		
  $\astlo'_t, \tilde \astlo_t$ and $\tilde \astlo'_t$ are defined in the natural way: namely, by replacing $\chi_t$ by respectively $ \chi_t', \tilde \chi_t$  and $\tilde \chi'_t$ in \eqref{eq:Ldefn}, while replacing $\eps d_{XY}$ by $s$, where $\chi_t'$ is given by %(see also \eqref{chits'}) %{eq:tildenotation}): %Here we used the natural generalization of \eqref{eq:chitdefn} 
		\begin{align}\label{chits'}	\chi'_{t}(|x|)=\chi'\left(\frac{|x|-R_{\min}(X)-v't}{s}\right).
	\end{align}

The proof of Theorem \ref{thm:timederivmaintxt} is lengthy and deferred to the supplemental material (SM). A key ingredient in the proof is the bound 
\begin{equation}\label{k}
%%c_{\mathrm{max}} \le
\| [J, \x]\|\le \ka^{(1)}_J \equiv v_{\max}%\ \text{ uniformly in }\ \Lam, 
\end{equation}
% of the special commutator $[J, \x]$  that arises naturally in the proof,  
(uniformly in $\Lam$) where $J f(x)= \sum_{y} J_{xy} f_y$  is an operator on the one-particle space $\ell^2(\Lambda)$. The bound \eqref{k} follows from Lemma \ref{lem:HRf-com} in the SM and the Schur test;
%  We recall that $\kappa_J^{(1)}$ defined in \eqref{J-cond} plays a special role as the upper bound on the speed of sound \eqref{eq:vmax} in the Bose-Hubbard model. 
it is where formula  \eqref{eq:vmax} for $v_{\max}$
 arises in our argument.

%\sout{Let us briefly discuss the key structural properties of the right-hand side of \eqref{eq:DPhiest}. First, notice that the leading and subleading terms are of the same form up to the replacements $f\to \tilde f$ and $\chi\to\tilde\chi$. Second, these terms appear with opposite signs (the operators $\astlo'_t$ and $\tilde \astlo'_t$ are positive because the derivatives $\chi'$ and $\tilde\chi'$ are positive). The upshot is that, after moving the leading term to the left-hand side, we can \textit{iterate the bound \eqref{eq:DPhiest}} as explained below, gaining a power of $s^{-1}$ with each iteration step.

\begin{proof}[Proof of Theorem \ref{thm:main}] 
%The idea is that, after moving the leading term to the left-hand side, we can iterate the bound \eqref{eq:DPhiest}, gaining a power of $s^{-1}$ with each step.

The key idea is to iterate \eqref{eq:DPhiest}. We fix $f\in\mathcal C_{\eta,\xi}$ and $\chi \in\mathcal C_{1/2,1}$. We use $s=\eps d_{XY}$, take the expectation of \eqref{eq:DPhiest} and integrate over time. Using that $\lan \Phi(t)\ran_t\geq 0$ and,  by \eqref{eq:Phiconditions}, $\lan \Phi(0)\ran_0= 0$, as well as $v'-v_{\max}=\eps v>0$ and $t\leq \tfrac{s}{\eps v}$, we obtain
$$
\int_0^t \lan f'(\astlo_r) \astlo'_r\ran_r \dd r\leq Cs^{-1} \int_0^t  \lan \tilde f'(\tilde \astlo_r) \tilde \astlo'_r\ran_r \dd r+Ct s^{1-p}.
$$
Since this holds for any $f\in\mathcal C_{\eta,\xi}$, we can iterate. It follows that there exist $\tilde f\in \mathcal C_{\eta,\xi}$ and $\tilde\chi\in\mathcal C_{1/2,1}$ so that
\begin{align}%$$\begin{aligned}
\label{ineq-iter}\int_0^t \lan f'(\astlo_r) \astlo'_r\ran_r  \dd r\leq& Cs^{1-p} \int_0^t  \lan\tilde f'(\tilde \astlo_r) \tilde \astlo'_r\ran_r \dd r+Ct s^{1-p}\notag\\
\leq& Ct s^{1-p}
\end{align}%\end{aligned}$$
where the second estimate uses that $\|\tilde f'(\tilde \astlo_r)\|\leq \|\tilde f'\|_\infty\leq C$ by the functional calculus and that $\lan\tilde \astlo_r'\ran_r\leq C$ which in turn follows from the Cauchy-Schwarz inequality $\lan b_x^\dagger b_y\ran_r\leq \lan n_x\ran_r+\lan n_y\ran_r$.

 Integrating the expectation of \eqref{eq:DPhiest} over time and using $\lan \Phi(t)\ran_t= \lan \Phi(r)\ran_r +\int_r^t \lan D\Phi(r)\ran_r \dd r$ and \eqref{ineq-iter},  we obtain, for any $t\ge r\ge 0$,
\beq \label{Phi-monot}  
	\lan \Phi(t)\ran_t	\leq \lan \Phi(r)\ran_r +C(t-r) s^{-p}, %\qedhere
\eeq 
showing the essential monotonicity of $\lan \Phi(t)\ran_t$ under the evolution. Setting here $r=0$ and using  \eqref{eq:Phiconditions} gives the desired bound $\lan P_{ \bar N_{Y}\ge \xi}\ran_{t}\leq  Ct s^{-p}$.
 \end{proof}
  
\section{Conclusions}
We have resolved a longstanding open problem in the area of quantum lattice gases by providing the first derivation of a maximal speed for macroscopic particle transport in the Bose-Hubbard model. Our result is a new kind of macroscopic-type Lieb-Robinson bound for particle transport. It complements other recent results \cite{WH,KS,YL,FLS} which hold for special initial states and are otherwise closer to the original formulation of the Lieb-Robinson bound.  %\sout{In the concurrent paper \cite{FLS}, we use a variant of the present techniques to prove a Lieb-Robinson bound on commutators of local observables in the BH model under a stronger assumption on the initial state.}

The central physical idea underpinning our proof is to engineer the ASTLOs, adiabatic and  spacetime observables whose support dynamically tracks and controls the surplus of particles outside the effective light cone and whose expectation values decrease under time evolution. %\st{Here  while $\Phi(t)\ge 0$, the Heisenberg derivative is essentially negative, $D\Phi(t)\le 0$ (up to unessential terms).} %localize the particles in spacetime \textit{adiabatically}. The adiabatic approach is instrumental in obtaining tight dynamical control on the spreading in spacetime. 

 The analytical method that we use is quite robust. For example, it applies without significant change to a wide variety of BH type models with different hoppings and different lattice structures. %We expect that the ASTLO approach extends to other quantum lattice gases with unbounded local interactions.

Regarding possible extensions, we note that our ASTLOs here are specifically designed to track particle transport and thereby naturally give rise to the commutator $[J,x]$. To control propagation of other 
physical quantities, e.g.\ entanglement, one would use adapted observables which have to satisfy the appropriate analog of \eqref{k} uniformly in $\Lambda$. This change would also affect the value of the maximal speed bound (but not its existence).

\section*{Acknowledgments}
The authors thank Tomotaka Kuwahara for useful comments on a draft version of the manuscript. They also thank Jens Eisert for informing them of related work currently under completion \cite{Eisert}. The research of IMS is supported in part by NSERC Grant No. NA7901.

 %%%%%%%%%% Merge with supplemental materials %%%%%%%%%%
 \widetext
 \pagebreak
 \begin{center}
 	\textbf{\large Supplemental Material:\\ Maximal speed for macroscopic particle transport in the Bose-Hubbard model}
 \end{center}
 %%%%%%%%%% Merge with supplemental materials %%%%%%%%%%
 %%%%%%%%%% Prefix a "S" to all equations, figures, tables and reset the counter %%%%%%%%%%

\stepcounter{myequation}
 \setcounter{figure}{0}
 \setcounter{table}{0}
 \makeatletter
 \renewcommand{\theequation}{S\arabic{equation}}
 \renewcommand{\thefigure}{S\arabic{figure}}
 %\renewcommand{\bibnumfmt}[1]{[S#1]}
 %\renewcommand{\citenumfont}[1]{S#1}
 %%%%%%%%%% Prefix a "S" to all equations, figures, tables and reset the counter %%%%%%%%%%
This appendix provides the complete proof of Theorem \ref{thm:timederivmaintxt}. % (in the slightly generalized form of Theorem \ref{thm:timederivmaintxt} below). %Afterwards, we include a short proof of self-adjointness of the Bose-Hubbard Hamiltonian (Proposition \ref{prop:sa}).
In the following, $c,C>0$ stand for generic positive constants whose value may change from line to line and which may implicitly depend on parameters such as $\|\chi'\|_\infty$ or on $\kappa_J^{(1)},\ldots,\kappa_J^{(p)}$ defined in \eqref{kaJp-def}.  %{J-cond}.
   %\begin{align}\label{eq:kappapdefn} \kappa_J^{(q)}=\max_{x\in\Lam}\sum_{y\in\Lam} |J_{xy}^\Lam| |x-y|^{q},\qquad q=1,\ldots,p. %\|J^\Lam\|:=\max_{y\in\Lam}\big(\sum_{x\in\Lam} (J_{xy}^\Lam)^2|x-y|^2\big)^{1/2} \end{align} 
Note that all $\kappa_J^{(q)}<\infty,  q=1,\ldots,p,$ by an assumption of Theorem \ref{thm:timederivmaintxt}. %\ref{J-cond}.

	\noindent\begin{figure}[h!]
\begin{center}
    \includegraphics[scale=.65]{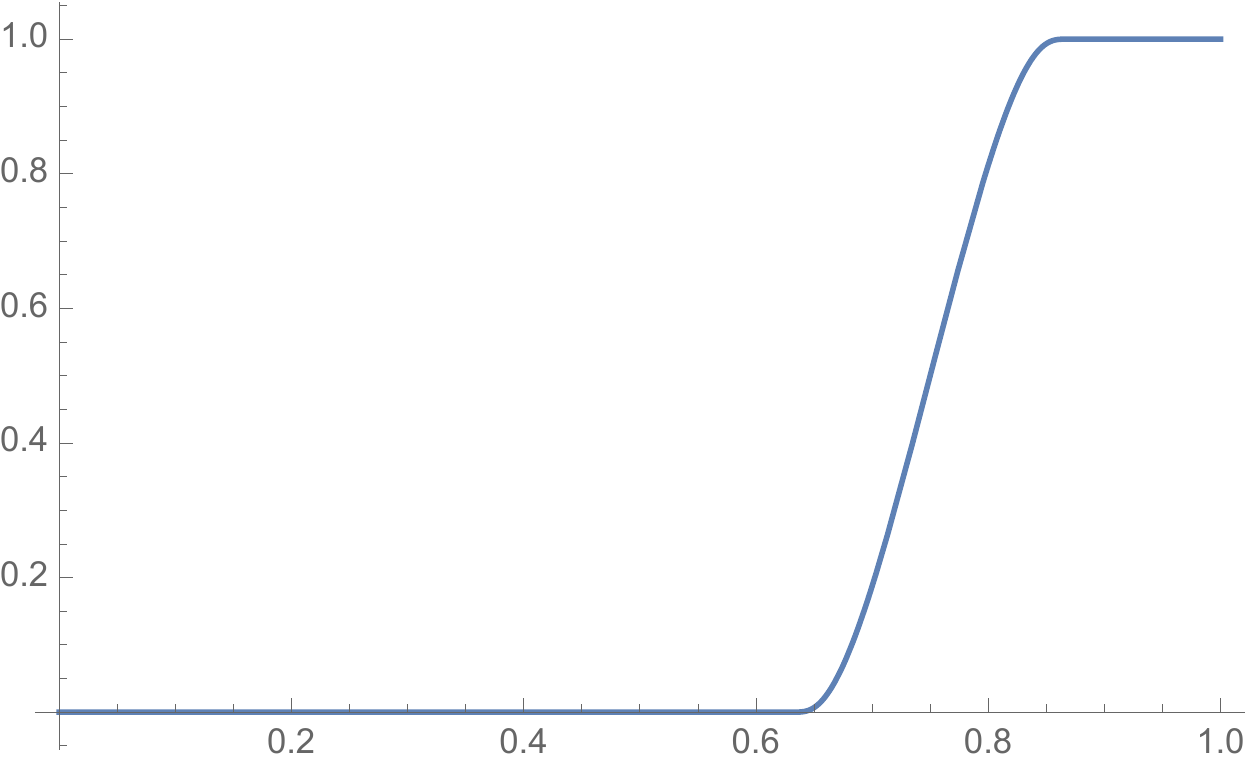}
    \caption{Example of a cutoff function $\chi\in \mathcal C_{1/2,1}$.}
    \label{fig1}
\end{center}
\end{figure}

Recall the definition of the set of cutoff functions,
\beq\label{Cetaxi}
\begin{aligned}
\mathcal C_{\eta,\xi}
=
\Big\{ 		
f\in C^\infty(\R_+)\,:\, f,f'\geq0,\, &\sqrt{f'}\in C^\infty(\R_+),
 f=0 \textnormal{ on } (0,\eta), f=1 \textnormal{ on } (\xi,\infty), \mathrm{supp}\, f'\subset (\eta,\xi)
\Big\}.
\end{aligned}
\eeq

An example of a cutoff function lying in $\mathcal C_{1/2,1}$ is shown in Figure \ref{fig1}. For $ \chi\in \mathcal C_{1/2,1}$, we write $\chi_{t,s}$ for \eqref{eq:chitdefn} with the variable $s$ replacing $\eps d_{XY}$, that is,
\beq\label{eq:chitsdefn}
\chi_{t,s}(|x|)=\chi\left(\frac{|x|-R_{\min}(X)-v't}{s}\right),
\eeq
and we define $\tilde \chi_{t,s}$ and $\chi'_{t,s}$ analogously; see eq.\ \eqref{chits'}. We also consider the generalized ASTLO
\beq\label{eq:astlotsdefn}
\astlo_{t,s}=\frac{1}{N} \dd\Gam(\chi_{t,s})
\eeq
Fix $f\in \mathcal C_{\eta,\xi}$. We shall consider the time evolution of the observable 
\begin{align}\label{propag-obs1}
	\Phi_s(t) & = f(\astlo_{t,s}).
\end{align}
The operators $\astlo_{t,s}', \tilde \astlo_{t,s}$ and $\tilde \astlo_{t,s}'$ are defined analogously as explained after Theorem \ref{thm:timederivmaintxt}.	
%
%\begin{theorem}\label{thm:timederivmaintxt} 
%Let $f\in \mathcal C_{\eta,\xi}$ and $\chi\in \mathcal C_{1/2,1}$. Then, there exists a constant $C>0$ and functions $\tilde f\in \mathcal C_{\eta,\xi}$ and $\tilde\chi\in \mathcal C_{1/2,1}$ such that for all $t$ and all sufficiently large $s$,
%	$$
%	D\Phi_s(t)\leq \frac{\kappa_J^{(1)}-v}{s} f'(\astlo_{t,s})\astlo'_{t,s}
%+Cs^{-2} \tilde f'(\tilde \astlo_{t,s})\tilde \astlo_{t,s}'+Cs^{-p},
%	$$
%\end{theorem}

 In the remainder of this section, we prove Theorem \ref{thm:timederivmaintxt} through various expansions in the small parameter $s^{-1}$.

\section{Toolbox and definitions} %{Preliminaries}
In this section, we prepare the proof of Theorem \ref{thm:timederivmaintxt} by recalling some mathematical tools used in the rigorous Schr\"odinger equation theory. %of Quantum Mechanics. %known facts about the Helffer-Sj\"ostrand formula and the second quantization functor. We also investigate basic properties of the function class $\mathcal C_{\eta,\xi}$.\\

%\paragraph{Notation.} For  a subset $U\subset \cL$, we write $U^c:=\Lam\setminus U$ for its complement in $\Lam$. We say an operator $A$ on our Fock space is supported in a set $U\subset \Lam$ and write $\supp A\subset U$, if it commutes with $b_x^\#$ for all $x\in U^c$. Note that $\supp A\cap \supp B=\emptyset$ implies $[A, B]=0$  and $\supp [A, B]=\supp A\cap \supp B$. Moreover, $\supp \one=\emptyset$ and $\supp 0=\emptyset$.

\subsection{%Review of 
Commutator expansions with error estimates}\label{sec:commut} 
We review relevant commutator expansions with error estimates. These results were first derived in \cite{SigSof} and then improved in \cite{Skib, HunSig1, HunSigSof}. We denote $\mathrm{ad}_{A} H=[A,H]$ and write $\mathrm{ad}_{A}^k$ for the $k$-fold iteration of this map.

We introduce the weighted norms $\|f \|_{m}=\int_\R (1+x^2)^{m/2} |f(x)|\dd x$.  
\begin{lemma}\label{lem:commut-exp} %
	Let $f\in C^\infty(\R)$ be bounded, 
	with $\sum_{k=0}^{M+2}\|f^{(k)}\|_{k-M-1}<\infty$, for some $M\ge 1$.
	Let  $A$ be a bounded self-adjoint operator and let $B$ be a bounded operator. 
	Then
	%, under the hypothesis of Theorem \ref{thm:min-vel-est}:
	\begin{align}\label{Hcomm-exp} &[B, f(A)]=  \sum_{k=1}^{M-1}\frac{1}{k!}\mathrm{ad}_{A}^k(B) f^{(k)}(A) +\mathrm{Rem}_M, \\
		&\textnormal{where } 		\label{Res-exp}\mathrm{Rem}_M(A,f)=\int_{\R^2} (z-A)^{-1}B_M(H)(z-A)^{-n}d\widetilde f(z).	
	\end{align}
	%$H^j B_k, j=0, 1, k =1, \dots, n-1,$ are bounded operators and $\|H^j O(s^{-n})\|\ls s^{-n}, j=0, 1$. 
	There exists a constant $C>0$ such that we have the error estimate
	\begin{equation}\label{eq:Remestimate}
		\| \mathrm{Rem}_M\| \le C  \| \mathrm{ad}_{A}^M(B)\| \sum_{k=0}^{M+2}\|f^{(k)}\|_{k-M-1}.
	\end{equation}

\end{lemma} 

Here and in the following, we use the convention that for $M=1$, the sum on the right-hand side of \eqref{Hcomm-exp} is omitted. 

%The following symmetrized version of Lemma \ref{lem:commut-exp} will occasionally be useful. It follows from the observation that $[H,f(A)]$ is anti-Hermitian and that $B_k f^{(k)}(H)$ is Hermitian/anti-Hermitian when $k$ is even/odd.
%
%\begin{corollary}
%	Under the conditions of Lemma \ref{lem:commut-exp}, we have
%	$$
%	[H,f(A)]= \sum_{\substack{1\leq k\leq M-1:\\ k \textnormal{ odd}}} \frac{1}{k!}B_k f^{(k)}(H)
%	+\frac{\mathrm{Rem}_M-\mathrm{Rem}_M^\dagger}{2}.
%	$$
%\end{corollary}

\begin{proof}[Proof of Lemma \ref{lem:commut-exp}] We only sketch the proof and refer to \cite{HunSig1} for the details. 
	The proof of Lemma \ref{lem:commut-exp} relies on the Helffer-Sj\"ostrand formula for a function $f$ of a self-adjoint operator $A$ and its derivatives, i.e.
	\begin{align} \label{fA-repr}
		&f^{(k)}(A)=k! \int_{\R^2} \dd\widetilde f(z)(z-A)^{-k-1}, \quad  \dd\widetilde f(z)=-\frac{1}{2\pi}\partial_{\bar z} \bar f(z) \dd x\dd y, %\qquad z=x+iy.
	\end{align}
where $ z=x+iy$ and $\widetilde f$ is an almost analytic extension of $f$.	We quote the following result from \cite{HunSig1}.
	
	\begin{lemma}[Lemma B.2 in \cite{HunSig1}]
		Let $M\geq 0$ and $f\in C^{M+2}(\R)$ with\\ $\sum_{k=0}^{M+1}\|f^{(k)}\|_{k-1}<\infty$. Then there exists an almost analytic extension $
		\widetilde f:\C\to\C$ of $f$ satisfying
		\begin{align} \label{tildef-est}
			&\int_{\R^2} |d\widetilde f(z)||\im(z)|^{-M-1}\le C \sum_{k=0}^{M+2}\|f^{(k)}\|_{k-M-1}
		\end{align}
		and \eqref{fA-repr} holds for all self-adjoint operators $A$. The integral in \eqref{fA-repr} converges in norm sense and is bounded uniformly in $A$.
	\end{lemma}
	
	The almost analytic extension $\tilde f$ can be defined in an explicit manner, see e.g.\ \cite[(B.5)]{HunSig1}.

	Using (B.14)-(B.15) and the remark following (B.18) of \cite{HunSig1}, we have
	\begin{align}
		\label{comm-exp1}[B, f(A)]&=\sum_{k=1}^{M-1}\frac{1}{k!}\mathrm{ad}_{A}^k(B) f^{(k)}(H)+\mathrm{Rem}_M.
		%\label{Bk-bnd}B_k&= \,{ad_{\x}^kg(H)};\\ %\hbox{ \it (bounded for every } k);\\
	\end{align}    
	We recall the convention that for $M=1$, the sum on the right-hand side is omitted. 
	%Now we show that the operators %$H^j B_k, j=0, 1, k =1, \dots, n,$ and  
	%$Re(s)$ are bounded.
	Since the operator $B_M$ is bounded, we can control the remainder via \eqref{tildef-est}, i.e.,
	\begin{align}\label{Res-est} 
		\| \mathrm{Rem}_M\|&\le  \| \mathrm{ad}_{A}^M(B)\|  \int_{\R^2} \|z-A\|^{-M-1} |d\widetilde f(z)|\\
		&\le  \| \mathrm{ad}_{A}^M(B)\|  \int_{\R^2} |\im z|^{-M-1} |d\widetilde f(z)|\\
		&		\le C \| \mathrm{ad}_{A}^M(B)\| \sum_{k=0}^{M+2}\|f^{(k)}\|_{k-M-1},
	\end{align}
	as desired.
\end{proof} 

\subsection{Basic properties of second quantization}
We begin by introducing some standard notation. Let us consider a one-particle operator $A:\ell^2(\Lam)\to\ell^2(\Lam)$, i.e., a $|\Lam|\times |\Lam|$ matrix $A$ acting as
\[
A f(x)=\sum_{y\in \Lam} A_{xy} f_y,\qquad f\in \ell^2(\Lam).
\]
We write $\mathrm{d}\Gamma(A)$ for its lift to the Fock space defined by
\begin{align} \label{Rk} 
\mathrm{d}\Gamma(A)=\sum_{x, y\in \Lam} b_x^\dagger A_{xy} b_y.\end{align}
We note that $\dd\Gamma$ is a linear map.

For instance, we can express the hopping term in the Hamiltonian \eqref{H-Lam} %, 
as  \begin{align} \label{T}
	T= \sum_{x, y\in \Lam} J_{xy} b_x^\dagger b_y=\dd\Gam(J),\ \text{ where  we set}\ %$k$ is defined by
 J f(x)= \sum_{y} J_{xy} f_y.\end{align} %In particular, 
It is convenient to abuse notation and to identify a function $F:\Lam\to\C$ with the  multiplication operator that acts diagonally on $f
\in \ell^2(\Lam)$ via $Ff(x)=F(x)f(x)$. Then
\begin{align} \label{R g} \dd\Gam(F)=\sum_{x\in \Lam}F(x) b_x^\dagger b_x
=\sum_{x\in \Lam}F(x) n_x.
\end{align} 
For instance, we can rewrite Definitions \eqref{eq:Nlocdefn} and \eqref{eq:Ldefn} as
\[
N_U=\dd\Gam(\mathbbm 1_U),\qquad \astlo_t=\tfrac{1}{N_\Lam}\dd\Gam(\chi_t)
\]

The canonical commutation relations for $b_x$ and $b_x^\dagger$ imply the following standard relation.
  \begin{align} \label{Rk-comm}
  	[\dd\Gam(A), \dd\Gam(B)]=\dd\Gam([A,B]).\end{align} 
In particular, for functions $F,G:\Lam\to\C$, we have that $ \dd\Gam(F)$ and $ \dd\Gam(G)$ commute.

Another general property of the second quantization is that it is monotonic with respect to the partial order on Hermitian operators. That is, for Hermitian $|\Lam|\times|\Lam|$ matrices $A$ and $B$, we have
\begin{equation}\label{eq:Gammon}
A\leq B\quad\Longrightarrow\quad \dd\Gam(A)\leq \dd\Gam(B).
\end{equation}
To verify \eqref{eq:Gammon}, we diagonalize $B-A=U\mathrm{diag}(\lam_1,\ldots,\lam_{|\Lam|})U^{-1}$ and exchange the order of summation to obtain
$$
\dd\Gam(B)-\dd\Gamma(A)=\sum_j \lam_j C_j^\dagger C_j,\qquad \textnormal{with } 
C_j=\sum_y \overline{U_{yj}} b_y.
$$

The following special case of an iterated commutator will be useful.

\begin{lemma}\label{lem:HRf-com} 
	Let $k\geq 1$ and $F:\Lam\to\C$. We have
\begin{equation}\label{HRf-com}
	\mathrm{ad}^k_{\dd\Gamma(F)}(H)
=\mathrm{ad}^k_{\dd\Gamma(F)}(T)=\dd\Gamma(\mathrm{ad}^k_F(J))
\end{equation}
where $\mathrm{ad}^k_F(J)$ is the $|\Lam|\times |\Lam|$ matrix with the matrix entries
\begin{align}
	\label{kf-com} & %[k, f]\ \text{  }\ 
\left(\mathrm{ad}^k_F(J)\right)_{xy}= J_{xy} (F(x)-F(y))^k,\qquad x,y\in\Lam.
\end{align} \end{lemma}

\begin{proof} 
The first relation in \eqref{HRf-com} follows from the fact that $\dd\Gamma(F)$ commutes with $H-T$ since both are linear combinations of the commuting opeators $n_x$.
The second relation in \eqref{HRf-com} follows from the fact that  
 $T=\dd\Gam(J)$ and the identity \eqref{Rk-comm}. 
Finally, \eqref{kf-com} holds by a straightforward induction.    \end{proof}

In particular, Lemma \ref{lem:HRf-com}  implies that %\begin{lemma}\label{lem:N-conserv} 
 the total particle number $N_\Lam$ is conserved: %For any subset $U\subset \Lam$, we have
 \begin{align} \label{N-conserv} [H_\Lam,  N_\Lam] =0.
  \end{align}

%\textit{Remark.}\,   A special commutator that arises naturally in the proof is $[J, \x]$.  It can be bounded through \eqref{kf-com} and the Schur test by
%\begin{equation}\label{k}
%%c_{\mathrm{max}} \le
 %\| [J, \x]\|\le \ka^{(1)}_J. \end{equation} We recall that $\kappa_J^{(1)}$ defined in \eqref{J-cond} plays a special role as the upper bound on the speed of sound in the Bose-Hubbard model. The bound \eqref{k} is how it arises in our argument.

\subsection{Admissible functions}
For an interval $I\subset \R$, we write $C_c^\infty\mathrm{(I)}$ for the class of smooth functions with compact support in $I$.
For the proof of Theorem \ref{thm:timederivmaintxt} , we introduce the following useful function class.

\begin{definition}\label{defn:admissible}
 Let $\xi,\eta\in [0,1]$ with $\eta<\xi$.
We introduce the class of admissible functions
$$
\mathcal A_{\eta,\xi}
=
\left\{ 		
h\in C_c^\infty((\eta,\xi))\,:\, h\geq0,\, \sqrt{h}\in C^\infty(\R)
\right\}.
$$
\end{definition}

\noindent\begin{figure}[t]
\begin{center}
    \includegraphics[scale=.65]{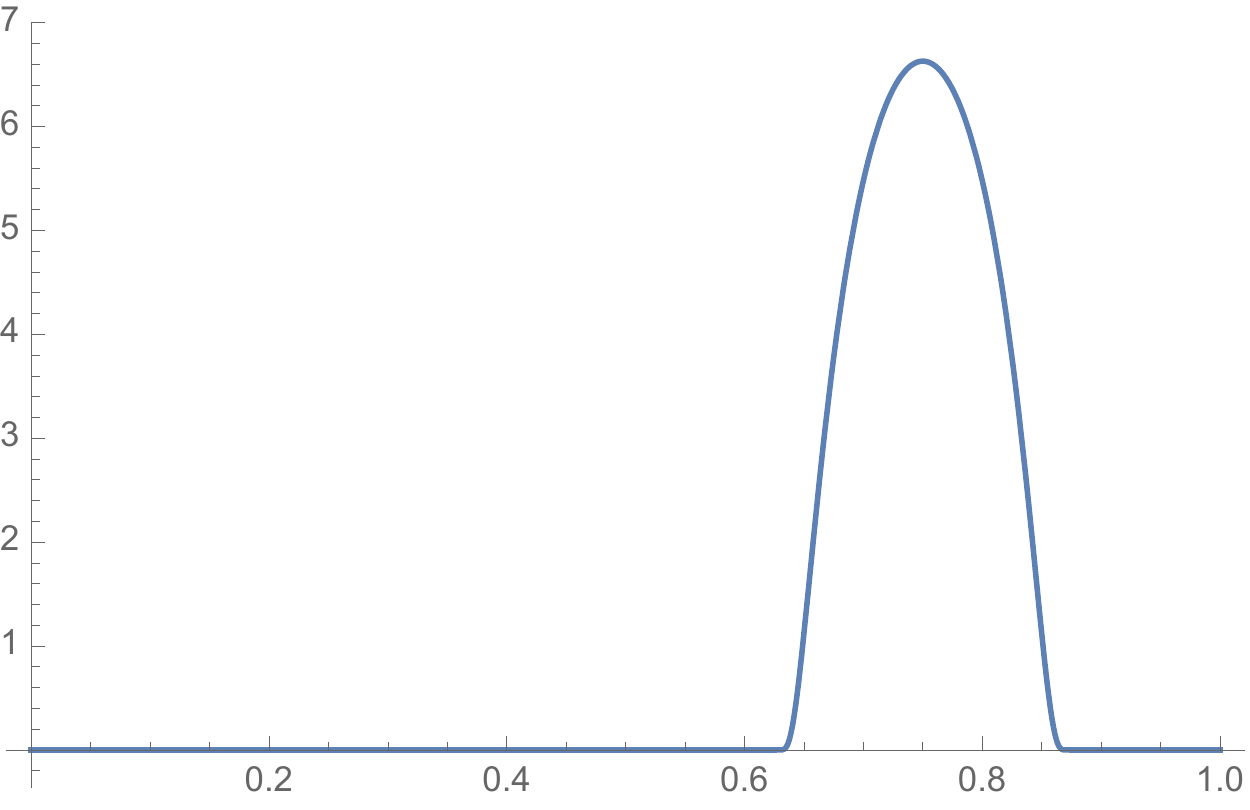}
    \caption{The derivative $h=\chi'$ of the cutoff function $\chi$ shown in Figure \ref{fig1}. Note that $h$ is an admissible function in $\mathcal A_{1/2,1}$.}
    \label{fig2}
\end{center}
\end{figure}

The following lemma shows that the elements of $\mathcal C_{\eta,\xi}$ from \eqref{Cetaxi}  can be seen as antiderivatives of admissible functions up to a multiplicative constant.  %By setting  \[f(r) =\frac{\int_{-\infty}^r 	h(\tilde r)\dd \tilde r}{\int_\R h(\tilde r)\dd \tilde r},\]we obtain the following
\begin{lemma}\label{lem:antideriv}
	%If $f\in \mathcal C_{\eta,\xi}$, then $f'\in \mathcal A_{\eta,\xi}$.
	If $h\in \mathcal A_{\eta,\xi}$, then there exists $f\in \mathcal C_{\eta,\xi}$ so that
	\[
	h(r)=f'(r) \int_\R h(\tilde r)\dd \tilde r
	\]
\end{lemma}

\begin{proof}
	%The first implication is trivial. For the second implication, 
	The lemma follows by setting  
	$$
	f(r) =\frac{\int_{-\infty}^r 	h(\tilde r)\dd \tilde r}{\int_\R h(\tilde r)\dd \tilde r}.
	$$
		\end{proof}

%%%%%%%%%%%%%%%%%%%%%%%

%We will denote
%\beq\label{eq:chi'tsdefn}
%\chi'_{t,s}(|x|)=\chi'\left(\frac{|x|-R_{\min}(X)-vt}{s}\right)
%\eeq

\subsection{Evolution of the propagation observables} %{Preliminary calculations}

 In this section, we calculate the Heisenberg derivative $D\Phi_s(t)$ defined in \eqref{Heis-der}.

For the first term in $D\Phi_s(t)$, cf.\ \eqref{Heis-der}, we have
\begin{align} \label{dt-Phi}
&{\partial\over{\partial t}}\Phi_s(t)=-\frac{v'}{s}\,  f^\prime(\astlo_{ts})\astlo_{ts}'.
\end{align}
with
\begin{equation}
\label{Rts'}
 	\astlo_{t,s}'=\frac{1}{N} \dd\Gam(\chi'_{t,s}),\ \text{ where $\chi'_{t,s}$ is defined in \eqref{chits'}.}\
	\end{equation}
Indeed, to verify \eqref{dt-Phi}, we note that $\astlo_{t,s}$ and $\astlo_{t,s}'$ commute and are both diagonal in the basis of Mott states \eqref{eq:mott}. On a given Mott state, \eqref{dt-Phi} then holds by the chain rule.
 
The main work is thus to consider the commutator $i[H, \Phi_s(t)]$ in \eqref{Heis-der}. Central objects in the argument are the multiple commutators:
\[
B_k=\mathrm{ad}_{\astlo_{t,s}}^k(i H),\qquad k\geq 1.
\]
We set $u_1=\sqrt{f'}$
which by  $f \in\mathcal{C}_{\eta,\xi}$ satisfies $u_1\geq 0$ and $u_1\in C_c^\infty((\eta,\xi))$. 
Furthermore, 
for $k\geqq 2$, %we use that $f^{(k)}\prec u_k$ to write $f^{(k)}=f^{(k)}u_k^2$ with $u_k\in C_c^\infty(\R_+)$.
we let $u_k\in C_c^\infty(\R_+)$ be s.t.  $f^{(k)}\prec u_k$, where 
we introduced the notation
\beq\label{eq:precdefn}  g_1\prec g_2\quad\stackrel{def}{\Longleftrightarrow}\quad g_2=1\textnormal{ on } \mathrm{supp}\, g_1.
 \eeq 
 %Then we can write $f^{(k)}=f^{(k)}u_k^2$.  
With these definitions, we have
\begin{lemma}\label{lm:HPhi-comm}
Assume $f\in \mathcal C_{\eta,\xi}$.  Let $\re A=\frac{i}2\big(A+A^\dagger\big)$. Then we have
\begin{align}\label{HPhi-comm}
&i[H, \Phi_s(t)]=u_1  B_1 u_1+S + R,\\
\label{S}&S=\sum_{k=2}^{p-1}  u_k \re \big( B_k f^{(k)}\big)u_k \notag\\ %(B_k g_k+(B_k g_k)^\dagger) u_k\\
&\hspace{2cm}+\sum_{k=1}^{p-1}\sum_{j=1}^{p-k-1} \frac{(-1)^j}{j!} \re \big(u_k^{(j)}(\astlo_{t,s}) B_{k+j} g_k u_k\big),\\
\label{R}&R=\re \bigg( \mathrm{Rem}_p(\astlo_{t,s},f)+ \sum_{k=1}^{p-1}\mathrm{Rem}_{p-k}(B_k,u_k)^\dagger g_k u_k\bigg), %\bigg(u_k^{(j)}(\astlo_{t,s}) B_{k+j} g_k u_k+\big(u_k^{(j)}(\astlo_{t,s}) B_{k+j} g_k u_k\big)^\dagger\bigg).
\end{align}
where $f^{(1)}\equiv f'$ and 
 $\mathrm{Rem}_p(A,f)$ is defined in \eqref{Res-exp}.\end{lemma}
\begin{proof}[Proof of Lemma \ref{lm:HPhi-comm}]
The assumption $f\in \mathcal C_{\eta,\xi}$ implies $f'\in C_c^\infty(\R_+)$. Hence we can apply Lemma \ref{lem:commut-exp} to $i[H, \Phi_s(t)]=i[H,f(\astlo_{t,s})]$ to obtain, for $p\ge 2$
$$
i[H, \Phi_s(t)] %=i[H,f(\astlo_{t,s})]
= \sum_{k=1}^{p-1} \frac{1}{k!} B_k f^{(k)}(\astlo_{t,s})+ \mathrm{Rem}_p(\astlo_{t,s},f). %\qquad M\geq 1.
$$
Next, we symmetrize this expression up to another commutator.

Defining %$g_1=$ and 
$g_k=f^{(k)}$ for $k\geq 1$, where $f^{(1)}\equiv f'$ and recalling  $f^{(k)}\prec u_k$, so that $f^{(k)}=f^{(k)}u_k^2$, we write 
$$
f^{(k)}=g_k u_k^2,\qquad k\geq 1.
$$
 In the following, for the sake of readability, we often suppress the argument $\astlo_{t,s}$ from the notation. We have
$$
B_kf^{(k)}
= u_k B_k g_k  u_k
+[B_k,u_k] g_ku_k,\qquad k\geq 1.
$$
%so that we have $g_k\prec u_k$ where
%$$
%g_k\prec u_k\quad\stackrel{def}{\Longleftrightarrow}\quad |g_k|=1 \textnormal{ on } \mathrm{supp}\, u_k.
%$$
The commutator $[B_k, u_k] \equiv [B_k,u_k(\astlo_{t,s})]$ can be further expanded via the adjoint version of Lemma \ref{lem:commut-exp},
$$
[B_k, u_k] %(\astlo_{t,s})]
=\sum_{j=1}^{p-1-k} \frac{(-1)^j}{j!}u_k^{(j)} %(\astlo_{t,s}) 
 B_{k+j} +\mathrm{Rem}_{p-k}(B_k, u_k)^\dagger.
$$

Combining these commutator expansions, we obtain
\begin{align}
\label{eq:commI+II}
i[H, \Phi_s(t)]&=\mathrm{(I)}+\mathrm{(II)}+\mathrm{(III)},\\
\label{eq:Idefn}
\mathrm{(I)}=& u_1 B_1 u_1,\\
\label{eq:IIdefn}
\mathrm{(II)}=& \sum_{k=2}^{p-1}  u_k B_k g_k u_k
+ \sum_{k=1}^{p-1}\sum_{j=1}^{p-k-1} \frac{(-1)^j}{j!}u_k^{(j)} B_{k+j} g_k u_k,\\
\label{eq:IIIdefn}
\mathrm{(III)}=& \mathrm{Rem}_p(\astlo_{t,s},f)+ \sum_{k=1}^{p-1}\mathrm{Rem}_{p-k}(B_k,u_k)^\dagger g_k u_k.
\end{align}
%Here $(S)$ stands for  term and $(R)$ stands for remainder term (at the order $p$).
Since $i[H, \Phi_s(t)]$ is self-adjoint, we have that $i[H, \Phi_s(t)]=\mathrm{(I)}+\re\big(\mathrm{(II)}\big)+\re\big(\mathrm{(III)}\big)$, which gives \eqref{HPhi-comm}. \end{proof}
%To summarize, from \eqref{Heis-der} and \eqref{dt-Phi}, we know that
%\beq\label{eq:thmtd1}
%D\Phi_s(t)=-\frac{v}{s}f'(\astlo_{t,s})\astlo_{t,s}+i[H,\Phi_s(t)]
%\eeq
\DETAILS{In view of \eqref{eq:Idefn} and \eqref{eq:IIdefn}, we have to provide operator inequalities on
$$
\begin{aligned}
&(I)=iu_1 B_1 u_1,\\
\frac{\mathrm{(II)}+\mathrm{(II)}^\dagger}{2}
&=\frac{i}{2}\sum_{k=2}^{p-1}  u_k (B_k g_k+(B_k g_k)^\dagger) u_k\\
&+\frac{i}{2}\sum_{k=1}^{p-1}\sum_{j=1}^{p-k-1} \frac{(-1)^j}{j!}\left(u_k^{(j)}(\astlo_{t,s}) B_{k+j} g_k u_k+\left(u_k^{(j)}(\astlo_{t,s}) B_{k+j} g_k u_k\right)^\dagger
\right)
\end{aligned}
$$
and $\mathrm{(III)}$. This equations together with \eqref{eq:commI+II} and \eqref{eq:IIIdefn} imply \eqref{HPhi-comm}.}

%%%%%%%%%%%%%%%%%
\section{Proof of Theorem \ref{thm:timederivmaintxt} } %{thm:timederivmaintxt}}{Estimate of  the r.h.s. the expansion \eqref{HPhi-comm}} 
 In the next subsections, we consider the symmetrized expansion \eqref{HPhi-comm} %{eq:commI+II} 
  and  estimate the three terms on the r.h.s.\ 
  in reverse order, starting with the norm bound on the remainder term $R$ which is the easiest.

%%%%%%%%%%%%%%%
\subsection{Controlling the remainder term $R$}
We first show that the remainder term $R$ in \eqref{HPhi-comm} is small as $s\to\infty$.

\begin{lemma}\label{lm:IIIest}
There exists a constant $C>0$ such that
$$
\|R\|\leq C s^{-p}\kappa^{(p)}_J,\qquad s\geq 1.
$$	
\end{lemma}

\begin{proof}[Proof of Lemma \ref{lm:IIIest}]
By the remainder estimate \eqref{eq:Remestimate}, we have
$$
\|\mathrm{Rem}_p(\astlo_{t,s},f)\|\leq  \| \mathrm{ad}_{\astlo_{t,s}}^p(H)\| \sum_{k=0}^{p+2}\|f^{(k)}\|_{k-p-1}
\leq C\|B_p\|.
$$
Similarly, using that $g_k\prec u_k$ and that $\|g_k(\astlo_{t,s})\|\leq \|g_k\|_\infty$ by the functional calculus,
\begin{align*}
\|\mathrm{Rem}_{p-k}(B_k,u_k)^\dagger g_k u_k\|
&\leq \|\mathrm{Rem}_{p-k}(B_k,u_k)\| \|g_k\|_\infty\\
& \leq \sum_{l=0}^{p-k+2}\|u_k^{(l)}\|_{l-p+k-1} \|B_p\|
\leq C \|B_p\|,\quad k\geq 1.
\end{align*}

We see that it remains to prove
\begin{equation}\label{eq:Bpestimate}
\|B_p\|\leq s^{-p}\kappa^{(p)}_J.
\end{equation} 

We recall that $\astlo_{t,s}=N^{-1}\dd\Gam(\chi_{t,s})$ and use Lemma \ref{lem:HRf-com}  with $F=\chi_{t,s}(|\cdot|)$ to write
\begin{align}\label{eq:Bprewrite}
B_p=\mathrm{ad}_{\astlo_{t,s}}^p(i H)&=\frac{1}{N}i \dd\Gamma(\mathrm{ad}^p_{\chi_{t,s}} (J))\notag\\
&=
i \frac{1}{N}\sum_{x,y\in\Lam} (\chi_{t,s}(|x|)-\chi_{t,s}(|y|))^p J_{xy} b_x^\dagger b_y.
\end{align}
Denote $\tilde B_k= i^{k-1} B_k=i^k\mathrm{ad}_{\astlo_{t,s}}^k( H)$. By applying the operator Cauchy-Schwarz inequality to the self-adjoint operators $i^p (\chi_{t,s}(|x|)-\chi_{t,s}(|y|))^p J_{xy} b_x^\dagger b_y$ and the symmetry $J_{yx}=J_{xy}$, we obtain
\begin{equation}\label{eq:BpestimateCS1}
\tilde B_p
\leq \frac{1}{N}\sum_{x,y\in\Lam} |\chi_{t,s}(|x|)-\chi_{t,s}(|y|)|^p |J_{xy}|n_x
\end{equation}

%Now by the Cauchy-Schwarz inequality on $\ell^2(\Lam\times \Lam)$ , we have
% \begin{equation}\label{eq:BpestimateCS2}
% \tilde B_p
% \leq \frac{1}{N}\sum_{x,y\in\Lam} |\chi_{t,s}(|x|)-\chi_{t,s}(|y|)|^p |J_{xy}| n_x
%\end{equation}
Finally the mean-value theorem implies
$$
|\chi_{t,s}(|x|)-\chi_{t,s}(|y|)|
\leq s^{-1} ||x|-|y|| \|\chi_{t,s}'\|_\infty
\leq s^{-1} C |x-y|
$$
and so $\tilde B_p \leq s^{-p} C \kappa^{(p)}_J$. This proves \eqref{eq:Bpestimate} %{eq:BMestimate}
 and hence Lemma \ref{lm:IIIest}.
\end{proof}

\subsection{Estimating the symmetrized subleading term $S$}
The argument used to prove Lemma \ref{lm:IIIest} can be refined if we replace the application of the mean-value theorem by iterated Taylor expansion. This is precisely what is needed for the subleading term $S$ in \eqref{HPhi-comm}.

%A key feature of the ASTLO construction is that we use smooth, slowly varying cutoff functions instead of sharp ones. We let $\chi:\R\to [0,1]$ be a smooth cutoff function. 
We recall that we assume that $\chi$ belongs to the following space of cutoff functions
\beq\label{eq:parentdefn}
\begin{aligned}
\mathcal C_{1/2,1}
=
\Big\{ 		
\chi\in &C^\infty(\R_+)\,:\, \chi,\chi'\geq0,\; \sqrt{\chi'}\in C^\infty(\R_+),\ \mathrm{supp}\, \chi'\subset (\tfrac{1}{2},1),\\
&\hspace{1.5cm}\chi(r)=0 \textnormal{ for }r\leq  \tfrac{1}{2},\; \chi(r)=1\textnormal{ for }r\geq 1
\Big\}.
\end{aligned}
\eeq

%We recall that $\chi$ is assumed to belong to the space of cutoff functions \eqref{eq:parentdefn},

\begin{proposition}\label{prop:IIest}
There exist a constant $C>0$ and functions $\tilde\chi\in\mathcal C_{1/2,1}$ , $h\in \mathcal{A}_{\eta,\xi}$ such that
\begin{align}\label{S-bnd}%\frac{\mathrm{(II)}+\mathrm{(II)}^\dagger}{2}
S \leq h(\tilde \astlo_{t,s}) \tilde \astlo'_{t,s} +Cs^{-p},
\end{align}
where, recall, the operators $\tilde \astlo_{t,s}$ and $\tilde \astlo_{t,s}'$ are defined in Theorem \ref{thm:timederivmaintxt}.	\end{proposition}

In the remainder of this subsection, we prove Proposition \ref{prop:IIest} in three separate steps. We begin by setting up convenient notation for Taylor expansions. Fix $1\leq k\leq p$. We can use Lemma \ref{lem:HRf-com}  with $F=\chi_{t,s}$ to write
\begin{equation}\label{eq:Bkrewrite}
B_k=\mathrm{ad}_{\astlo_{t,s}}^k(i H)=i \frac{1}{N}\dd\Gamma(\mathrm{ad}^k_{\chi_{t,s}} (J))
\end{equation}
 Therefore the main object we aim to control is the iterated commutator 
\beq\label{eq:matrielt}
(\mathrm{ad}^k_{\chi_{t,s}} (J))_{xy}
=(\chi_{t,s}(|x|)-\chi_{t,s}(|y|))^k J_{xy}.
\eeq

By Taylor's theorem with Lagrange remainder, we have the option to expand for any $L\geq 0$
\begin{align*}
\chi_{t,s}(|x|)-\chi_{t,s}(|y|)
&=\sum_{\ell=1}^{L-1} \frac{\chi_{t,s}^{(\ell)}(|x|)}{\ell!} (|x|-|y|)^\ell 
+ R_L\\
&=\sum_{\ell=1}^{L-1} s^{-\ell} \frac{(\chi^{(\ell)})_{t,s}(|x|)}{\ell!} (|x|-|y|)^\ell 
+ R_L,
\end{align*}
with the remainder bound 
$|R_L|\leq s^{-L} \frac{\|\chi^{(L)}\|_\infty}{L!} ||x|-|y||^L
\leq s^{-L} C |x-y|^{L}$.  It is convenient to introduce the notation
\begin{align}\label{eq:chitaylor}
\chi_{t,s}(|x|)&-\chi_{t,s}(|y|)
=\sum_{\ell=1}^L\mathcal T^{(L)}_\ell,\\  %\quad 
 \notag&\textnormal{with } \mathcal T_l^{(L)}=
\begin{cases}
s^{-\ell} \frac{(\chi^{(\ell)})_{t,s}(|x|)}{\ell!} (|x|-|y|)^\ell,
& \textnormal{for } 1\leq \ell\leq L-1\\
R_L, \quad &\textnormal{for } \ell=L.
\end{cases}
\end{align} 
We note that all terms in the expansion satisfy a bound of the form
\beq\label{eq:mathcalTlbound}
|\mathcal T^{(L)}_\ell|\leq C s^{-\ell} |x-y|^\ell,\qquad 1\leq\ell\leq L,
\eeq
where the constant $C$ only depends on $\ell$ and $\chi_{t,s}$.\\

\textbf{Step 1: Symmetrically preserving support information.}
We introduce localizing functions on the left and right side of the Hermitian matrix $i^k  (\mathrm{ad}^k_{\chi_{t,s}} (J))$. This symmetric sandwiching is needed for proving an operator inequality of the form \eqref{S-bnd}. %{eq:propchiclaim}.

\begin{lemma}\label{lm:step1localize}
There exist constants $c,C>0$ and a function $\tilde\chi\in\mathcal C_{1/2,1}$ such that
\beq\label{eq:lmstep1localize}
i^k  (\mathrm{ad}^k_{\chi_{t,s}} (J))
=c \sqrt{\tilde\chi'_{t,s}} i^k  (\mathrm{ad}^k_{\chi_{t,s}} (J)) \sqrt{\tilde\chi'_{t,s}}
+\mathcal R,
\eeq
where $\mathcal R$ is a Hermitian matrix satisfying the norm bound
\beq\label{eq:Rnormbound}
\|\mathcal R\|\leq C \kappa_J^{(p)} s^{-p}.
\eeq
\end{lemma}

%We remark that $\chi_k$ implicitly depends on $M_0$ as well, but this dependence will play no role in the remainder.

\begin{proof}[Proof of Lemma \ref{lm:step1localize}]
Recall \eqref{eq:precdefn}. We choose $\tilde u\in C_c^\infty(\tfrac{1}{2},2)$ with $\tilde u\geq 0$, such that
\beq\label{eq:ucondn}
\chi^{(1)},\chi^{(2)},\ldots,\chi^{(p)}\prec \tilde u.
\eeq
Then we define $\tilde\chi\in\mathcal C_{1/2,1}$ by
\beq\label{eq:tildechidefn}
\tilde\chi(r) =\frac{1}{c}\int_{-\infty}^{r} 	\tilde u(r')^2\dd r',\qquad c=\int_\R \tilde u( \rho)^2\dd \rho.
\eeq
(Compare the proof of Lemma \ref{lem:antideriv}.) We have $\sqrt{c\tilde\chi'}=\tilde u$ and hence also $\sqrt{c\tilde\chi_{t,s}'}=\tilde u_{t,s}$. 

The matrix $\mathcal R$ can now be written as
\[
\mathcal R=i^k \mathrm{ad}^k_{\chi_{t,s}} (J)-
i^k \tilde u_{t,s}  \mathrm{ad}^k_{\chi_{t,s}} (J)\tilde u_{t,s}.
\]
We note that $\mathcal R$ is automatically Hermitian as the difference of two Hermitian matrices and so it suffices to prove the norm bound \eqref{eq:Rnormbound}. For this, we consider a fixed $(x,y)$-matrix element $\mathcal R_{xy}$ which by \eqref{eq:matrielt} reads
$$
\mathcal R_{xy}=i^k  (\chi_{t,s}(|x|)-\chi_{t,s}(|y|))^{k} (1-\tilde u_{t,s}(|x|)\tilde u_{t,s}(|y|)) J_{xy}.
$$
We decompose
$$
1-\tilde u_{t,s}(|x|)\tilde u_{t,s}(\y)=1-\tilde u_{t,s}(\x)+\tilde u_{t,s}(\x)(\tilde u_{t,s}(\x)-\tilde u_{t,s}(\y)).
$$

We first consider the term $1-\tilde u_{t,s}(\x)$ and employ a Taylor expansion of order $L=p-k+1$ to obtain

\begin{align*}
(1-\tilde u_{t,s}(\x))&(\chi_{t,s}(|x|)-\chi_{t,s}(|y|))^k\\
=& (1-\tilde u_{t,s}(\x))\left(\sum_{\ell=1}^{p-k+1} \mathcal T_\ell^{(p-k+1)}\right)
(\chi_{t,s}(|x|)-\chi_{t,s}(|y|))^{k-1}\\
=&(1-\tilde u_{t,s}(\x)) R_{p-k+1}
(\chi_{t,s}(|x|)-\chi_{t,s}(|y|))^{k-1}
\end{align*}
 where we used $(1-\tilde u)\chi^{(\ell)}=0$ for all $1\leq\ell\leq p$. %We note that $T_\ell(u)^{(p-k+1)}$ is defined as in \eqref{eq:chitaylor} with $u$ replacing $\chi$.
 
 By \eqref{eq:mathcalTlbound}, we can bound the absolute value of this expression by
 $$
 C s^{-p}|x-y|^{-p}.
 $$
 Taylor expanding around the point $y$ instead yields the same bound, albeit with a potentially different constant $C$, on the second term $\tilde u_{t,s}(\x)(\tilde u_{t,s}(\x)-\tilde u_{t,s}(\y))$.
 
 By the Schur test and the fact that $\mathcal R$ is Hermitian, we obtain the norm bound
 \beq\label{R-bnd} \|\mathcal R\|\leq
 \sup_{x\in\Lam} \sum_{y\in\Lam} |R_{xy}|
 \leq Cs^{-p} \sup_x \sum_{y\in\Lam} |x-y|^{-p} |J_{xy}|
 =Cs^{-p} \kappa_J^{(p)}
 \eeq
and Lemma \ref{lm:step1localize} is proved.
\end{proof}

\textbf{Step 2: Bound on the iterated commutator.}
In this step, we prove 
\begin{lemma}\label{lm:IIest}
	There exist constants $c,C>0$ and a function $\tilde\chi\in\mathcal C_{1/2,1}$ such that for every $1\leq k\leq p$, the iterated commutators are bounded as
	\beq\label{eq:Bkformbound}
	\pm \tilde B_k \equiv \pm i^k\mathrm{ad}_{\astlo_{t,s}}^k( H)\leq s^{-k} c \tilde \astlo'_{t,s}
	+C s^{-p}.
	\eeq
where, recall, the operator $\tilde \astlo_{t,s}'$ is defined in Theorem \ref{thm:timederivmaintxt} and is given by (cf. \eqref{chits'}) 
	\beq\label{eq:tildenotation}
\begin{aligned}
\tilde \astlo'_{t,s}=&\frac{1}{N} \dd\Gam(\tilde\chi'_{t,s}),\qquad \tilde \chi'_{t,s}(|x|)=&\tilde\chi'\left(\frac{|x|-R_{\min}(X)-vt}{s}\right).
\end{aligned}
\eeq
\end{lemma}

\begin{proof}[Proof of Lemma \ref{lm:IIest}]
We observe that it suffices to prove the operator inequalities 
\begin{equation}\label{eq:propchiclaim}
\pm i^k  \mathrm{ad}^k_{\chi_{t,s}} (J)\leq s^{-k} c \tilde \chi'_{t,s}+Cs^{-p}.
\end{equation}
Indeed, assuming \eqref{eq:propchiclaim}, the monotonicity and linearity of second quantization $\dd\Gamma(\cdot)$, see \eqref{eq:Gammon}, give
\begin{align*}
\pm \tilde B_k=\frac{1}{N}\dd\Gamma(\pm i^k  \mathrm{ad}^k_{\chi_{t,s}} (J))
&\leq s^{-k} \frac{c}{N} \dd\Gamma(\tilde \chi'_{t,s})+\frac{C}{N}s^{-p}\dd\Gamma(\mathbbm 1)\\
&\leq s^{-k}c \tilde \astlo'_{t,s}
+Cs^{-p}.
\end{align*} the last step used that $\dd\Gamma(\mathbbm 1)= N$.

We shall prove the following norm bound
\beq\label{eq:normboundremains}
\|\mathrm{ad}^k_{\chi_{t,s}} (J)\|\leq C  s^{-k}.
\eeq
This will imply the modified claim \eqref{eq:propchiclaim}. Indeed, together Lemma \ref{lm:step1localize} and \eqref{eq:normboundremains} give
\begin{align*}
\pm i^k  (\mathrm{ad}^k_{\chi_{t,s}} (J))
=&c \sqrt{\tilde\chi'_{t,s}} i^k  (\pm \mathrm{ad}^k_{\chi_{t,s}} (J)) \sqrt{\tilde\chi'_{t,s}}+\mathcal R\\
\leq& c\|\mathrm{ad}^k_{\chi_{t,s}} (J)\| \tilde \chi_{t,s}'+\|\mathcal R\|\\
\leq& c \tilde\chi_{t,s}'+C s^{-p} \kappa^{(p)}_J
\end{align*}
up to a change of the constant $c>0$.

We now prove \eqref{eq:normboundremains}. We shall use \eqref{eq:chitaylor} but need to be careful when expanding $(\chi_{t,s}(x)-\chi_{t,s}(y))^k$ because we can only control overall polynomial powers up to order $|x-y|^{p}$ through $\kappa_J^{(p)}$. Therefore, we iteratively expand only as far as necessary to get the desired error $s^{-p}$.

The iterative Taylor expansion reads
\begin{align*}
&(\chi_{t,s}(x)-\chi_{t,s}(y))^k=\sum_{\ell_1=1}^{p-k+1}\mathcal T^{(p-k+1)}_{\ell_1}
\sum_{\ell_2=1}^{p-k+2-\ell_1}\mathcal T^{p-k+2-\ell_1)}_{\ell_2}\\
&\times \sum_{\ell_3=1}^{p-k+3-\ell_1-\ell_2}\mathcal T^{(p-k+2-\ell_1-\ell_2)}_{\ell_3}
\ldots 
\sum_{\ell_k=1}^{p-\ell_1-\ldots-\ell_{k-1}}\mathcal T^{(p-\ell_1-\ldots-\ell_{k-1})}_{\ell_k},
\end{align*}
where $\mathcal T^{(L)}_\ell$ are defined after \eqref{eq:chitaylor} and with the usual convention that empty sums equal zero. As can be seen from the last term, the orders of the Taylor expansions are chosen so that any admissible tuple $(\ell_1,\ldots,\ell_k)$ satisfies $k\leq \ell_1+\ldots+\ell_k\leq p$.

The estimate \eqref{eq:mathcalTlbound} implies
$$
\mathcal T^{(p-k+1)}_{\ell_1}\mathcal T_{\ell_2}^{(p-k+2-\ell_1)}\ldots \mathcal T_{\ell_k}^{(p-\ell_1-\ldots-\ell_{k-1})}
\leq C s^{-\ell_1-\ldots-\ell_k}|x-y|^{\ell_1+\ldots+\ell_k}.
$$
By the Schur test,
\begin{align*}
\|\mathrm{ad}^k_{\chi_{t,s}} (J)\|
\leq& 
C \sup_{x\in\Lam} \sum_{y\in\Lam}|J_{xy}|
\sum_{\ell_1, \ldots, \ell_k}  s^{-\ell_1-\ldots-\ell_k}|x-y|^{\ell_1+\ldots+\ell_k}\\
\leq& C \sum_{p=k}^{p} s^{-p} \kappa_J^{(p)}
\leq C s^{-k}, 
\end{align*}
where $\sum_{\ell_1, \ldots, \ell_k}=\sum_{\ell_1=1}^{p-k}
\sum_{\ell_2=1}^{p-k+1-\ell_1}\sum_{\ell_3=1}^{p-k+2-\ell_1-\ell_2}
\ldots 
\sum_{\ell_k=1}^{p+1-\ell_1-\ldots-\ell_{k-1}}$, which yields \eqref{eq:normboundremains} and hence Lemma \ref{lm:IIest}.
\end{proof}

\textbf{Step 3: Addressing asymmetry and concluding Proposition \ref{prop:IIest}}
 
While Lemma \ref{lm:IIest} goes in the right direction, it is not so obvious how to use it to obtain an operator inequality for $S$ because in \eqref{S}, $B_k$ does not appear in the symmetric form $C^\dagger B_k C$. 

In our specific situation, the asymmetry can be addressed by combining the following two technical observations. 

\begin{itemize}
\item[(i)] Any operator inequality $A\leq B$ with $B> 0$ can be rephrased as the norm bound $\|B^{-1/2}AB^{-1/2}\|\leq 1$ and in our situation the target observable $\tilde \astlo_{t,s}'=N^{-1}\dd\Gamma(\tilde\chi'_{t,s})\geq 0$ is positive semidefinite and can thus be made positive definite by a limiting procedure.
\item[(ii)]
The target observable $\tilde \astlo_{t,s}'$ commutes with the source of the asymmetry, $g_k=g_k(\astlo_{t,s})$ and the latter is uniformly bounded by the functional calculus, $\|g_k(\astlo_{t,s})\|\leq \|g_k\|_\infty\leq c$. 
 \end{itemize}
 The details are as follows.

\begin{proof}[Proof of Proposition \ref{prop:IIest}]
Recall \eqref{eq:precdefn}. Fix $1\leq k\leq p$ and find $v\in C_c^\infty((\eta, \xi))$ with $v\geq0$ so that
\[
\left\{u_k^{(j)}\right\}_{\substack{2\leq k\leq p-1\\
0\leq j\leq p-1}
}\prec v.
\]

We claim that
\beq\label{eq:IIclaim}
S\leq c\tilde \astlo'_{t,s}+Cs^{-p}.
\eeq

This will be sufficient to conclude the lemma. Indeed, it implies
\begin{align*}
S
=v(\astlo_{t,s})S v(\astlo_{t,s})
&\leq v(\astlo_{t,s})(c\tilde \astlo'_{t,s}+Cs^{-p})v(\astlo_{t,s})\\
&\leq c\ h(\astlo_{t,s}) \tilde \astlo'_{t,s}+Cs^{-p}.
\end{align*}
where we defined the admissible function $h=v^2\in\mathcal{A}_{\eta,\xi}$.

It remains to prove the claim \eqref{eq:IIclaim}. A generic term contributing to $S$ is of the form
\begin{align*}
\re\big( w_1(\astlo_{t,s}) B_k w_2(\astlo_{t,s})\big), % &+ (i w_1(\astlo_{t,s}) B_k w_2(\astlo_{t,s}))^\dagger\\&=w_1(\astlo_{t,s}) B_k w_2(\astlo_{t,s})-w_2(\astlo_{t,s}) B_k^\dagger w_1(\astlo_{t,s})
\end{align*}
where $w_1,w_2$ are real-valued functions. (For example,  $w_1=u_k$ and and $w_2=g_k u_k$ gives $u_k (B_kg_k+g_kB_k^\dagger)u_k$.) In the following, we shall again suppress the argument $\astlo_{t,s}$ from the notation.

Let $\varepsilon>0$. We claim that 
\beq\label{eq:Bgclaimeps}
\re\big(w_1 B_k w_2\big) %+(w_1B_kw_2)^\dagger
\leq \frac12 cs^{-k}(\tilde \astlo_k'+\varepsilon)+Cs^{-p}
\eeq
with the constant $C>0$ as in Lemma \ref{lm:IIest} and $c>0$ to be determined. This implies \eqref{eq:IIclaim} by sending $\eps\to0$.

We would like to derive \eqref{eq:Bgclaimeps} via Lemma \ref{lm:IIest}. As mentioned before, the main challenge is to address the asymmetry due to $w_1\neq w_2$. %We distinguish cases based on the parity of $k$. 

We will derive \eqref{eq:Bgclaimeps} from
\beq\label{eq:Bgclaimepsw1w2}
\re\big(w_1  B_kw_2\big) %w_1 B_kw_2+(w_1B_kw_2)^\dagger
\leq \frac12 cs^{-k}(\tilde \astlo_k'+\varepsilon)+ Cs^{-p}w_1w_2
\eeq
by using that $\|w_1(\astlo_{t,s})\|\|w_2(\astlo_{t,s})\|\leq \|w_1\|_\infty\|w_2\|_\infty\leq c$ thanks to the functional calculus.

Since $\tilde \astlo_k'+\varepsilon> 0$ and $w_1w_2=w_2w_1$, the claim \eqref{eq:Bgclaimepsw1w2}  is equivalent to the norm bound
\beq\label{w1Bkw2-est}
\left\|D%\left(
\re\big(w_1( B_k-Cs^{-p} )w_2\big) %+\left(w_1(B_k-Cs^{-p} )w_2\right)^\dagger\right)
D
\right\|\leq 2cs^{-k},
\eeq
where $D=\frac{1}{\sqrt{\tilde \astlo_{t,s}'+\varepsilon}}$.
To estimate the left-hand side we use the commutativity
\beq\label{eq:commutator0}
\left[D, w_j(\astlo_{t,s})\right]=0,\qquad j=1,2,
\eeq
by the functional calculus and $[\tilde \astlo_{t,s}',\astlo_{t,s}]=N^{-1}\dd\Gamma([\tilde\chi'_{t,s},\chi_{t,s}])=0$. This allows us to pull out the norms of $w_1$ and $w_2$. Using this,  the estimate $\|w_1(\astlo_{t,s})\|\|w_2(\astlo_{t,s})\|\leq \|w_1\|_\infty\|w_2\|_\infty\leq c$ and the relation $\|A^\dagger\|=\|A\|$, we obtain
\begin{align*}
%&\left\|\frac{1}{\sqrt{\tilde \astlo_{t,s}'+\varepsilon}}\left( w_1(B_k-i^{-k}Cs^{-p})w_2 +\left(w_1(B_k-i^{-k}Cs^{-p} )w_2\right)^\dagger\right)\frac{1}{\sqrt{\tilde \astlo_{t,s}'+\varepsilon}}\right\|\\ \leq& \left\|\frac{1}{\sqrt{\tilde \astlo_{t,s}'+\varepsilon}}(B_k-i^{-k} Cs^{-p})  \frac{1}{\sqrt{\tilde \astlo_{t,s}'+\varepsilon}}\right\|\|w_1(\astlo_{t,s})\|\|w_2(\astlo_{t,s})\|\\\leq& c\left\|\frac{1}{\sqrt{\tilde \astlo_{t,s}'+\varepsilon}}(B_k-Cs^{-p})  \frac{1}{\sqrt{\tilde \astlo_{t,s}'+\varepsilon}}\right\|.
&\big\| D\re\big(w_1( B_k-Cs^{-p})w_2\big) D\big\|\\
\leq& \big\|D( B_k-Cs^{-p}) D\big\|\|w_1(\astlo_{t,s})\|\|w_2(\astlo_{t,s})\|\\
\leq& c\big\|D( B_k-Cs^{-p})  D\big\|.
\end{align*}
%where the estimate used that $\|w_1(\astlo_{t,s})\|\|w_2(\astlo_{t,s})\|\leq \|w_1\|_\infty\|w_2\|_\infty\leq c$ by the functional calculus.
%Since $k$ is even, we have $B_k=\pm i^kB_k$ for an appropriate choice of the sign $\pm$. 
Since $k\le p$,  the triangle inequality and Lemma \ref{lm:IIest} give, up to changing the constant $c$,
the inequality
$$
\left\|D(B_k-Cs^{-p})  D\right\|
\leq cs^{-k}
$$
which proves \eqref{w1Bkw2-est} and therefore \eqref{eq:Bgclaimepsw1w2}. 
 \end{proof}

%\st{when $k$ is even.} \textit{Case 2: $k$ is odd.} In this case we actually prove the stronger estimate
%\beq\label{eq:Bgclaimeps'}
%w_1i^{k} B_kw_2+(w_1i^{k}B_kw_2)^\dagger\leq cs^{-k}(\tilde \astlo_k'+\varepsilon)
%\eeq
%which trivially implies \eqref{eq:Bgclaimeps}.
%
%The claim \eqref{eq:Bgclaimeps'} is equivalent to the norm bound
%\begin{align*}
%\bigg\|M\re\bigg(
%w_1(i^{k}B_k-Cs^{-p} )w_2
%%+\big(w_1(B_k-i^{-k}Cs^{-p} )w_2\big)^\dagger
%\bigg)
%&M
%\bigg\|\leq cs^{-k}.
%\end{align*}
%Here the contributions from $i^{-k}Cs^{-p}$ cancel out because $w_2^\dagger w_1^\dagger=w_1w_2$ and $(i^{-k})^\dagger=-i^{-k}$ for $k$ odd. This estimate follows by using the commutativity \eqref{eq:commutator0} and Lemma \ref{lm:IIest} in the same way as for $k$ even.

\subsection{Estimating the main term $iu_1 B_1 u_1$}
We can estimate the leading term term $iu_1 B_1 u_1$ in a more refined way compared to $S$ by using that the first derivative has a sign, $\chi'\geq 0$. This fact allows to reproduce $
\astlo'_{t,s}$ exactly at lowest order (in favor of the $\tilde \astlo'_{t,s}$ that appeared above for higher orders) as asserted in Theorem \ref{thm:timederivmaintxt} . 

\begin{lemma}\label{lm:Iest}
Let $\tilde\chi\in\mathcal C_{1/2,1}$ be given by Lemma \ref{lm:step1localize}. There exists a constant $C>0$ such that 
\begin{align}
\mathrm{(I)}\leq \kappa_J^{(1)}  s^{-1}  \astlo'_{t,s}+Cs^{-2}\tilde \astlo'_{t,s}+Cs^{-p},\qquad p\geq 3.
\end{align}
\end{lemma}

%\begin{remark}
%For readers interested in small values of $p$, we mention that by stopping the expansion in the upcoming proof of Lemma \ref{lm:Iest} at second order, one can obtain the same result for $p= 2$ with the weaker error estimate $Cs^{-3/2}$, which ultimately leads to a weaker decay estimate outside of the light cone.
%\end{remark}

\begin{proof}[Proof of Lemma \ref{lm:Iest}]
%By Definition \eqref{eq:Idefn}, we have $\mathrm{(I)}=iu_1 B_1 u_1$. 
Since $u_1=u_1^\dagger$ appears symmetrically and $\dd\Gamma(\cdot)$ is monotonic, it suffices to prove the operator inequality
\beq\label{eq:sufficesop1}
i\mathrm{ad}_{\chi_{t,s}}(J)\leq \kappa_J^{(1)}  s^{-1}\chi'_{t,s}+Cs^{-2}\tilde \chi'_{t,s}+Cs^{-p}
\eeq
By applying Lemma \ref{lm:step1localize} with $k=1$, there exist constants $c,C>0$ and $\tilde\chi\in \mathcal C_{1/2,1}$ such that
\beq\label{eq:Ilocalize}
i\mathrm{ad}_{\chi_{t,s}}(J)
=c \sqrt{\tilde\chi'_{t,s}} i  \mathrm{ad}_{\chi_{t,s}} (J) \sqrt{\tilde\chi'_{t,s}}
+\mathcal R
\eeq
where the remainder $\mathcal R$ is Hermitian with norm controlled by $\|\mathcal R\|\leq Cs^{-p}$ (see \eqref{R-bnd}) and can thus be ignored in the following. Moreover, the construction in the proof of Lemma \ref{lm:step1localize} satisfies the relation
$$
c\tilde\chi'_{t,s}=1,\qquad \textnormal{on } \mathrm{supp}\, \chi,
$$
as can be seen from \eqref{eq:ucondn} and \eqref{eq:tildechidefn}.

Combining this with \eqref{eq:Ilocalize}, we see that \eqref{eq:sufficesop1} is implied by the operator inequality
\beq\label{eq:Iestclaimrephrase}
i\mathrm{ad}_{\chi_{t,s}}(J)\leq \kappa_J^{(1)} s^{-1}\chi'_{t,s}+Cs^{-2}
\eeq

Similarly to Step 3 in the proof of Proposition \ref{prop:IIest}, we rephrase the claimed operator inequality \eqref{eq:Iestclaimrephrase} as the following norm bound,\beq\label{eq:Iestrephrasenorm}
\left\|
M'i\mathrm{ad}_{\chi_{t,s}}(J)
M'
\right\|\leq 1,\ \text{ where }\  M'=1\big/\sqrt{\kappa_J^{(1)} s^{-1}\chi'_{t,s}+Cs^{-2} }.\eeq
%where $M':=\frac{1}{\sqrt{\kappa_J^{(1)} s^{-1}\chi'_{t,s}+Cs^{-2} }}$.

We shall prove \eqref{eq:Iestrephrasenorm} via the Schur test. We first consider the matrix elements 
$$
|i  (\mathrm{ad}_{\chi_{t,s}} (J))_{xy}|
=
|\chi_{t,s}(|x|)-\chi_{t,s}(|y|)||J_{xy}|.
$$
We consider a mixture of the Taylor expansions around $x$ and around $y$. This can in fact be extended to any order; see \cite[Lemma 2.2]{FLS}. 

Without loss of generality, assume $|x|\geq |y|$. By monotonicity, we have $\chi_{t,s}(|x|)\geq \chi_{t,s}(|y|)$ and by \eqref{eq:chitaylor} 
\begin{align*}
|\chi_{t,s}(|x|)-\chi_{t,s}(|y|)|
=\chi_{t,s}(|x|)&-\chi_{t,s}(|y|)
=\chi'_{t,s}(|x|) \frac{|x|-|y|}{s}+R_2^{(2)}\\
&\leq \chi'_{t,s}(|x|) \frac{|x-y|}{s}+C s^{-2}|x-y|^2.
\end{align*}
Since $\chi\in \mathcal C_{1/2,1}$, we have $u=\sqrt{\chi'}\in C^\infty(\R)$ and so
$$
\chi'_{t,s}(|x|)\leq 
u_{t,s}(|x|)u_{t,s}(|y|)+Cs^{-1}|x-y|
$$
We have shown that for $|x|\geq |y|$,
$$
|\chi_{t,s}(|x|)-\chi_{t,s}(|y|)|\leq s^{-1}\sqrt{\chi'_{t,s}(|x|)\chi'_{t,s}(|y|)}|x-y|+Cs^{-2}|x-y|^2.
$$
The same estimate holds if $|x|\geq |y|$ by interchanging the roles of $x$ and $y$ in the above argument.

Hence, we can bound the matrix elements appearing in \eqref{eq:Iestrephrasenorm} by
\begin{align*}
&\left|\left(M'i\mathrm{ad}_{\chi_{t,s}}(J)
M'\right)_{xy}\right|\\
\leq&
M'\left(\sqrt{\chi'_{t,s}(|x|)\chi'_{t,s}(|y|)}\frac{|x-y|}{s}
+C s^{-2} |x-y|^3\right) %\\&\qquad \times
|J_{xy}| M'\\
\leq&|J_{xy}|\frac{|x-y|}{c}+|J_{xy}||x-y|^3.
\end{align*}
 Applying the Schur test and recalling the Definition \eqref{kaJp-def} of $\kappa_J^{(p)}$ proves \eqref{eq:Iestrephrasenorm} and thus Lemma \ref{lm:Iest}.
\end{proof}

\subsection{Conclusion of the proof of Theorem \ref{thm:timederivmaintxt} }
\begin{proof} 

\DETAILS{We consider the symmetrized expansion \eqref{eq:commI+II} and %estimate the remainder by its norm 
using $\mathrm{(III)}+\mathrm{(III)}^\dagger \leq \|\mathrm{(III)}+\mathrm{(III)}^\dagger\|\leq \|\mathrm{(III)}\|+\|\mathrm{(III)}^\dagger\|=2\|\mathrm{(III)}\|$ and the fact that  (I) is Hermitian, we obtain
\beq\label{eq:thmtd}
\begin{aligned}
i[H,\Phi_s(t)]=&\frac{i[H,\Phi_s(t)]+(i[H,\Phi_s(t)])^\dagger}{2}\\
\leq& (I)+\frac{\mathrm{(II)}+\mathrm{(II)}^\dagger}{2}
+ \|\mathrm{(III)}\| 
\end{aligned}
\eeq
By combining \eqref{Heis-der}, \eqref{dt-Phi}, and \eqref{eq:thmtd}, we have
$$
D\Phi_s(t)\leq - \frac{v'}{s} f'(\astlo_{t,s})\astlo'_{t,s}+\mathrm{(I)}+\frac{\mathrm{(II)}
+\mathrm{(II)}^\dagger}{2}+\|\mathrm{(III)}\|
$$}
We apply estimates of $R, S$ and $u_1iB_1u_1$ given in Lemma \ref{lm:IIIest}, Propositions \ref{prop:IIest} and Lemma \ref{lm:Iest} to the r.h.s.\ of  the expansion \eqref{HPhi-comm} to obtain
\beq
D\Phi_s(t)\leq \frac{\kappa_J^{(1)}-v'}{s} f'(\astlo_{t,s})\astlo'_{t,s}
+Cs^{-2} h(\astlo_{t,s})\tilde \astlo_{t,s}'+Cs^{-M}
\eeq 
 By Lemma \ref{lem:antideriv}, there exists $\tilde f\in \mathcal C_{\eta,\xi}$ such that
$$
h=C \tilde f'.
$$
Recall \eqref{eq:precdefn}. We can find $\hat\chi\in \mathcal C_{\eta,\xi}$ satisfying
$$
\chi,\tilde\chi\prec \hat\chi
$$
and so, by monotonicity of $\tilde f$,
$$
C\tilde f'(\astlo_{t,s}) \tilde \astlo'_{t,s}\leq C\tilde f'(\hat \astlo_{t,s}) \hat \astlo'_{t,s}
$$
Finally, we rename $\hat\chi$ as $\tilde\chi$ again to avoid confusion with the Fourier transform. This proves Theorem \ref{thm:timederivmaintxt} .
\end{proof}

%%%%%%%%%%%%%%%%%%
\section{Remark on parameter dependencies}

%\begin{remark}
It is in principle possible to obtain the dependence of implicit constants on the parameters $\eta,\xi$ in Theorems \ref{thm:main} and \ref{thm:timederivmaintxt}. This could be used, to widen the scope of our result to situations of mesoscopic particle transport, i.e., propagation of a total of $\propto N^{\delta}$ particles with $0<\delta<1$. 

For this, one simply takes $\eta=\eta_0 N^{\delta-1}$ and $\xi=\xi_0 N^{\delta-1}$. In that case,  the distance $d_{XY}$ in \eqref{max-vel-estHub'} will be accompanied by a factor $N^{\delta-1}$. This means that the final estimate is useful on sufficiently large scales compared to the total particle number.

While we do not track the precise dependence of the constants $\eta,\xi$ for the sake of simplicity, we explain here how this can be done in principle. The key observation is that $\xi,\eta$ can be removed from the function classes $\mathcal A_{\eta,\xi}$ and $\mathcal C_{\eta,\xi}$ by an affine change of variables 
$$
a (r)=(2\xi-2\eta)r+2\eta-\xi
$$
which sends $[\tfrac{1}{2},1]\to [\eta,\xi]$. We have
$$
\begin{aligned}
\mathcal A_{\eta,\xi}=&
\left\{
h\in C^\infty(\R)\, :\, h(r)=h_1(a(r)) \textnormal{ with } h_1\in \mathcal A_{\tfrac{1}{2},1}\right\},\\
\mathcal C_{\eta,\xi}=&
\left\{
f\in C^\infty(\R)\, :\, f(r)=f_1(a(r)) \textnormal{ with } f_1\in \mathcal C_{\tfrac{1}{2},1}
\right\}
\end{aligned}
$$
In particular, for $h\in \mathcal A_{\eta,\xi}$ it holds that
$$
|h^{(k)}(r)|=2^k(\xi-\eta)^k |h_1^{(k)}(a(r))|
$$
where $h_1\in \mathcal A_{\tfrac{1}{2},1}$ does not explicitly depend on $\eta,\xi$. An analogous statement holds for $f\in\mathcal C_{\eta,\xi}$. We see that each derivative is naturally accompanied by a factor $\xi-\eta$ in addition to the factor $s^{-1}$ that arose for each derivative taken in the proof of Theorem \ref{thm:timederivmaintxt} given in the preceding sections.

\end{document}